\newcommand{\myabstract}[0]{
A square root approach is considered for the problem of accounting for model noise
in the forecast step of the ensemble Kalman filter (EnKF) and related algorithms.
The primary aim is to replace the method of simulated, pseudo-random, additive noise
so as to eliminate the associated sampling errors.
The core method is based on the analysis step of ensemble square root filters,
and consists in the deterministic computation of a transform matrix.
The theoretical advantages regarding dynamical consistency are surveyed,
applying equally well to the square root method in the analysis step.
A fundamental problem due to the limited size of the ensemble subspace is discussed,
and novel solutions that complement the core method are suggested and studied.
Benchmarks from twin experiments with simple, low-order dynamics
indicate improved performance over standard approaches
such as additive, simulated noise and multiplicative inflation.
}

\documentclass[10pt,a4paper]{article}

\def\abstract#1{\def\theabstract{%
\centerline{\vtop{\footnotesize\hsize 5.125in
\noindent\hskip8pt\relax #1\vskip2em}}}}

\abstract{\myabstract{}}

\def\acknowledgments{\paragraph*{Acknowledgments.}}

\long\def\appendixtitle#1{\section{#1}}

\def\topline{\hline\hline\vrule height 10pt depth4pt width0pt\relax}
\def\midline{\hline\vrule height 10pt width0pt\relax}
\def\botline{\hline}

\usepackage{dblfloatfix}

\usepackage{setspace}

\usepackage[top=2.5cm,bottom=2.5cm,left=1.5cm,right=1.5cm]{geometry}

\usepackage{layouts}

\usepackage{multicol}

\usepackage[small,bf]{caption}

\usepackage[T1]{fontenc}
\usepackage{lmodern}

\usepackage{endnotes}

\usepackage{paralist}            
\usepackage[ampersand]{easylist} 

\usepackage[usenames,dvipsnames]{xcolor}

\usepackage{amsfonts,amssymb}    
\usepackage[intlimits]{amsmath}  
\usepackage{amsthm}              
\usepackage{mathrsfs}            
\usepackage{bm}                  
\usepackage{upgreek}             
\usepackage{dsfont}              
\usepackage{mathtools}           
\usepackage{relsize}             
\usepackage{cancel}              

\usepackage{nicefrac}

\DeclarePairedDelimiter\abs{\lvert}{\rvert}%
\DeclarePairedDelimiter\norm{\lVert}{\rVert}%
\makeatletter
\let\oldabs\abs
\def\abs{\@ifstar{\oldabs}{\oldabs*}}
\let\oldnorm\norm
\def\norm{\@ifstar{\oldnorm}{\oldnorm*}}
\makeatother

\usepackage{float}      %
\usepackage{subfig}     %
\usepackage{fancyvrb}		
\usepackage{wrapfig} 		
\usepackage{tabularx}   
\usepackage{booktabs}   
\usepackage{cases}

\newtheoremstyle{mystyle}
  {\topsep}   
  {\topsep}   
  {\itshape}  
  {0pt}       
  {\bfseries} 
  {.}         
  {\newline}  
	{\thmname{#1}\thmnumber{ #2}\textnormal{\thmnote{ -- #3}}} 

\theoremstyle{mystyle}
\newtheorem{theo}{Theorem}[]
\newtheorem{lemm}{Lemma}[]
\newtheorem{coro}{Corollary}[]

\definecolor{mycc}{rgb}{0, 0, 0.45}
\usepackage{ifpdf}
\ifpdf
	\usepackage[pdftex=true,
		pdftitle={ExtSqrt},
		pdfauthor={Patrick Nima Raanes},
		hyperindex=true,
		colorlinks=true,
		linkcolor=blue,
		citecolor=mycc]{hyperref}
	\usepackage{epstopdf}

	\usepackage[stretch=10]{microtype}
\else
	\usepackage[hypertex,
		hyperindex=true,
		colorlinks=false]{hyperref}
	\usepackage{graphicx}
	\usepackage{srcltx}
\fi

\graphicspath{{./}{Pics/}}

\hypersetup{
	allbordercolors=violet 
}

\usepackage[colon,sort,authoryear]{natbib}

\usepackage{cleveref}
\crefname{equation}{eqn.}{eqns.}
\Crefname{equation}{Eqn.}{Eqns.}
\crefname{subeqns}{eqns.}{eqns.}
\Crefname{subeqns}{Eqns.}{Eqns.}
\crefformat{subeqns}{eqns.~(#2#1#3)}
\Crefformat{subeqns}{Eqns.~(#2#1#3)}
\crefname{figure}{Fig.}{Figs.}%
\Crefname{figure}{Figure}{Figures}%
\crefname{equation}{eqn.}{eqns.}%
\Crefname{equation}{Equation}{Equations}%
\crefname{table}{Table}{Tables}%
\Crefname{table}{Table}{Tables}%
\crefname{theo}{Theorem}{Theorems}
\Crefname{theo}{Theorem}{Theorems}
\crefname{theo}{Theorem}{Theorems}
\Crefname{theo}{Theorem}{Theorems}
\crefname{lemm}{Lemma}{Lemmas}
\Crefname{lemm}{Lemma}{Lemmas}
\crefname{defi}{Definition}{Definitions}
\Crefname{defi}{Definition}{Definitions}
\crefname{prop}{Proposition}{Propositions}
\Crefname{prop}{Proposition}{Propositions}
\crefname{property}{Property}{Properties}
\Crefname{property}{Property}{Properties}
\crefname{coro}{Corollary}{Corollaries}
\Crefname{coro}{Corollary}{Corollaries}
\crefname{algocf}{Algorithm}{Algorithms}
\Crefname{algocf}{Algorithm}{Algorithms}

\usepackage{authblk}

\newcommand{\NERSC}{Nansen Environmental and Remote Sensing Center}
\newcommand{\myemail}{patrick.raanes@maths.ox.ac.uk}

\title{Extending the square root method \\ to account for additive forecast noise in ensemble methods}

\author[1,2]{Patrick Nima Raanes\thanks{
	}}
\author[2]{Alberto Carrassi}
\author[2]{Laurent Bertino}
\affil[1]{Mathematical Institute, University of Oxford, OX2 6GG, UK}
\affil[2]{\NERSC, Thorm{\o}hlensgate 47, Bergen, N-5006, Norway}

\newcommand{\trsign}{{\mathrm{T}}}
\newcommand{\tr}{\ensuremath{^{\trsign}}}
\newcommand{\trsqrt}{\ensuremath{^{\trsign/2}}}

\newcommand{\pinvsign}{{+}}

\newcommand{\pinv}{\ensuremath{^\pinvsign}}

\newcommand{\pinvtr}{\ensuremath{^{\pinvsign \trsign}}}

\newcommand{\mat}[1]{{\pmb{\mathsf{#1}}}}
\newcommand{\bvec}[1]{{\bm{#1}}}

\newcommand{\x}[0]{\bvec{x}}
\newcommand{\y}[0]{\bvec{y}}
\newcommand{\q}[0]{\bvec{q}}

\newcommand{\bu}[0]{\bvec{u}}
\newcommand{\bv}[0]{\bvec{v}}
\newcommand{\vA}[0]{\bvec{v}_{A}}
\newcommand{\vB}[0]{\bvec{v}_{B}}
\renewcommand{\d}[0]{\bvec{d}}
\newcommand{\br}[0]{\bvec{r}}
\newcommand{\bmu}[0]{{\bvec{\mu}}}
\newcommand{\Q}[0]{\mat{Q}}
\newcommand{\hQ}[0]{\mat{\hat{Q}}}
\newcommand{\Qhh}[0]{\mat{\hat{Q}}{}^{1/2}}
\newcommand{\Qhtr}[0]{\mat{\hat{Q}}{}\trsqrt}
\newcommand{\QQh}[0]{[\Q - \hQ]}
\newcommand{\bPi}[0]{\mat{\Pi}}
\newcommand{\PiQ}[0]{\bPi}
\newcommand{\PiA}[0]{\bPi_{\A}}

\newcommand{\A}[0]{\mat{A}}
\newcommand{\E}[0]{\mat{E}}
\newcommand{\U}[0]{\mat{U}}
\newcommand{\V}[0]{\mat{V}}
\newcommand{\Sig}[0]{\mat{\Sigma}}
\newcommand{\Omeg}[0]{\mat{\Omega}}

\newcommand{\bL}[0]{\mat{L}}
\newcommand{\G}[0]{\mat{G}}
\newcommand{\T}[0]{\mat{T}}
\newcommand{\C}[0]{\mat{C}}

\newcommand{\bS}[0]{\mat{S}}
\newcommand{\bss}[0]{\bvec{s}}

\newcommand{\I}[0]{\mat{I}}

\newcommand{\bK}[0]{\mat{\bar{K}}}
\newcommand{\bP}[0]{\mat{P}}
\newcommand{\barP}[0]{\mat{\bar{P}}}
\newcommand{\F}[0]{\mat{F}}
\newcommand{\bH}[0]{\mat{H}}
\newcommand{\R}[0]{\mat{R}}
\newcommand{\Z}[0]{\mat{Z}}
\newcommand{\D}[0]{\mat{D}}
\newcommand{\hD}[0]{\mat{\hat{D}}}
\newcommand{\tD}[0]{\mat{\tilde{D}}}
\newcommand{\Dobs}[0]{\mat{D}_{\text{obs}}}

\newcommand{\M}[0]{\mat{M}}
\newcommand{\bXi}[0]{\mat{\Xi}}
\newcommand{\hXi}[0]{\mat{\hat{\Xi}}}
\newcommand{\tXi}[0]{\mat{\tilde{\Xi}}}
\newcommand{\bxi}[0]{\bvec{\xi}}
\newcommand{\hxi}[0]{\bvec{\hat{\xi}}}
\newcommand{\txi}[0]{\bvec{\tilde{\xi}}}

\newcommand{\iinds}[0]{i=1:m}
\newcommand{\xinds}[0]{0,1,\ldots}
\newcommand{\yinds}[0]{1,2,\ldots}
\newcommand{\tinds}[0]{t=\xinds}
\newcommand{\tindsy}[0]{t=\yinds}
\newcommand{\ninds}[0]{n=1:N}
\newcommand{\kMax}[0]{25}

\newcommand{\bx}[0]{\bvec{\bar{x}}}

\newcommand{\compactN}[0]{{N\!-\!1}}
\newcommand{\AN}[0]{{(\I_N - \ones\ones\tr/N)}}
\newcommand{\RMSE}[0]{{\text{RMSE}}}

\newcommand{\qpar}[0]{\q}
\newcommand{\qperp}[0]{\q^{\indep}}
\newcommand{\xipar}[0]{\bxi}
\newcommand{\xiperp}[0]{\bxi^{\indep}}

\newcommand{\SRD}[0]{\textsc{Sqrt-Dep}}

\newcommand{\SRA}[0]{\textsc{Sqrt-Add-Z}}

\newcommand{\SRN}[0]{\textsc{Sqrt-Core}}
\newcommand{\ADD}[0]{\textsc{Add-Q}}
\newcommand{\MULT}[0]{\textsc{Mult}-$m$}
\newcommand{\MULS}[0]{\textsc{Mult}-$1$}
\newcommand{\TSVD}[0]{\textsc{T-SVD}}

\newcommand{\SRDb}[0] {\textbf{\textsc{Sqrt-Dep}}}
\newcommand{\SRAb}[0] {\textbf{\textsc{Sqrt-Add-Z}}}
\newcommand{\SRNb}[0] {\textbf{\textsc{Sqrt-Core}}}

\newcommand{\bmr}[0]{\bvec{r}}
\newcommand{\barQ}[0]{\mat{\bar{Q}}}

\newcommand{\Lambd}[0]{\mat{\Lambda}}

\newcommand{\tn}[1]{{\textnormal{{#1}}}}

\newcommand{\Reals}{\mathbb{Re}}
\newcommand{\Imags}{i\Reals}

\newcommand{\Expect}[0]{\mathop{}\! \mathbb{E}}

\newcommand{\NormDist}{\mathcal{N}}

\newcommand{\Var}[0]{\mathop{}\! \mathbb{V}\tn{ar}}

\newcommand{\indep}[0]{\perp\!\!\!\perp} 

\DeclareMathOperator*{\diag}{diag}

\DeclareMathOperator*{\rank}{rank}

\DeclareMathOperator*{\trace}{trace}

\DeclareMathOperator*{\vspan}{span}

\newcommand{\ones}[0]{\mathds{1}}

\newcommand*\pdf{\mathop{}\! p}
\newcommand*\diff{\mathop{}\! d}

\newcommand{\dd}[2]{\frac{\diff #1}{\diff #2}}

\newcommand{\pd}[2]{\frac{\mathrm{\partial}{#1}}{\mathrm{\partial}{#2}}}

\newcommand{\dt}{\Delta t \,}
\newcommand{\dtObs}{\Delta t_{\text{obs}} \,}

\begin{document}

\twocolumn[
  \begin{@twocolumnfalse}
    \maketitle
		\theabstract
  \end{@twocolumnfalse}
]
{
  \renewcommand{\thefootnote}%
    {\fnsymbol{footnote}}
  \footnotetext[1]{\myemail \\
		Permission to place a copy of this work on this server has been provided by the
		AMS. The AMS does not guarantee that the copy provided here is an accurate copy
		of the published work.
}

\section{Introduction}
\label{sec:Introduction}

The EnKF is a popular method for doing data assimilation (DA) in the geosciences.
This study is concerned with the treatment of model noise in the EnKF forecast step.

\subsection{Relevance and scope}
\label{sec:Relevance_and_scope}
While uncertainty quantification is an important end product of any estimation procedure,
it is paramount in DA due to the sequentiality
and the need to correctly weight the observations at the next time step.
The two main sources of uncertainty in a forecast are the initial conditions and model error
\citep{slingo2011uncertainty}.
Accounting for model error is therefore essential in DA.

Model error, the discrepancy between nature and computational model,
can be due to incomplete understanding,
linearisation, truncation, sub-grid-scale processes, and numerical imprecision
\citep{nicolis2004dynamics,li2009accounting}.
For the purposes of DA, however, model error is frequently described
as a stochastic, additive, stationary, zero-centred, spatially correlated, Gaussian white noise process.
This is highly unrealistic, yet defensible in view of the multitude of unknown error sources,
the central limit theorem, and tractability \citep[][\S 3.8]{jazwinski1970stochastic}.
Another issue is that the size and complexity of geoscientific models makes it infeasible
to estimate the model error statistics to a high degree of detail and accuracy,
necessitating further reduction of its parameterisations
\citep{dee1995line}.

The model error in this study adheres to all of the above assumptions.
This, however, renders it indistinguishable from a noise process,
even from our omniscient point of view.
Thus, this study effectively also pertains to
natural noises not generally classified as model error, such as
inherent stochasticity (e.g. quantum mechanics)
and stochastic, external forcings (e.g. cosmic microwave radiation).
Therefore, while model error remains the primary motivation,
model \emph{noise} is henceforth the designation most used.
It is left to future studies to recuperate more generality by scaling back on the assumptions.

Several studies in the literature are concerned with the estimation of model error,
as well as its treatment in a DA scheme
\citep{daley1992estimating,zupanski2006model,mitchell2014carassi}.
The scope of this study is more restricted, addressing the treatment only.
To that end, it is functional to assume that the noise statistics,
namely the mean and covariance, are perfectly known.
This unrealistic assumption is therefore made,
allowing us to focus solely on the problem of incorporating
or \emph{accounting for} model noise in the EnKF.

\subsection{Model noise treatment in the EnKF}
\label{sec:Survey}

From its inception, the EnKF has explicitly considered model noise and accounted for it
in a Monte-Carlo way: adding simulated, pseudo-random noise
to the state realisations \citep{evensen1994sequential}.
A popular alternative technique is multiplicative inflation,
where the spread of the ensemble is increased by some ``inflation factor''.
Several comparisons of these techniques exist in the literature
\citep[e.g.][]{hamill2005accounting,whitaker2008ensemble,deng2011evaluation}.

Quite frequently, however, model noise is not explicitly accounted for,
but treated simultaneously with other system errors,
notably sampling error and errors in the specification of the noise statistics
\citep[][]{whitaker2004reanalysis,hunt2004four,houtekamer2005atmospheric,anderson2009spatially}.
This is because (a) inflation can be used to compensate for these system errors too,
and (b) tuning separate inflation factors seems wasteful or even infeasible.
Nevertheless, even in realistic settings, it can be rewarding to treat model error explicitly.
For example, \citet{whitaker2012evaluating} show evidence that,
in the presence of multiple sources of error,
a tuned combination of a multiplicative technique and additive noise
is superior to either technique used alone.

\Cref{sec:Alternative_approaches} discusses the EnKF model noise incorporation techniques
most relevant to this manuscript.
However, the scope of this manuscript is not to provide a full comparison
of all of the alternative under all relevant circumstances,
but to focus on the square root approach.
Techniques not considered any further here
include using more complicated stochastic parameterisations \citep{arnold2013stochastic,berry2014linear},
physics-based forcings such as stochastic kinetic energy backscatter \citep{shutts2005kinetic},
relaxation \citep{zhang2004impacts},
and boundary condition forcings.

\subsection{Framework}
\label{sec:Framework}
Suppose the state and observation, $\x^t \in \Reals^m$ and $\y^t \in \Reals^p$ respectively,
are generated by:
\begin{align}
	\label{eqn:state_sde_discrete}
	\x^{t+1} &= f(\x^t) + \q^t \, , && \tinds \, , \\
	\label{eqn:obs_eqn}
	\y^t &= \bH \x^t + \bmr^t \, , && \tindsy \, ,
\end{align}
where the Gaussian white noise processes $\{\q^t \mid \tinds \}$ and $\{\br^t \mid \tindsy \}$,
and the initial condition, $\x^0$, are specified by:
\begin{align}
	\q^t \sim \NormDist(0,\Q) \, ,
	\;
	\br^t \sim \NormDist(0,\R) \, ,
	\;
	\x^0 \sim \NormDist(\bmu^0, \bP^0) \, .
\end{align}
The observation operator, $\bH \in \Reals^{p \times m}$,
has been assumed linear because that is how it will effectively be treated anyway
\citep[through the augmentation trick of e.g.][]{anderson2001ensemble}.
The parameter $\bmu^0 \in \Reals^m$ is assumed known, as are the
symmetric, positive-definite (SPD) covariance matrices $\bP^0, \Q \in \Reals^{m^2}$ and
$\R \in \Reals^{p^2}$.
Generalisation to time-dependent $\Q,\R,f$, and $\bH$ is straightforward.

Consider $\pdf(\x^t \mid \y^{1:t})$,
the Bayesian probability distribution of $\x^t$
conditioned on all of the previous observations, $\y^{1:t}$,
where the colon indicates an integer sequence.
The recursive filtering process is usually broken into two steps:
the forecast step, whose output is denoted by the superscript $f$, and the analysis step,
whose output is denoted using the superscript $a$.
Accordingly, the first and second moments of the distributions are denoted
\begin{alignat}{2}
	\label{eqn:KF_notation_P}
	\x^f &= \Expect(\x^t|\y^{1:t-1}) \, ,  &\qquad \bP^f &= \Var(\x^t|\y^{1:t-1}) \, , \\
	\label{eqn:KF_notation_x}
	\x^a &= \Expect(\x^t|\y^{1:t}) \, ,    &\qquad \bP^a &= \Var(\x^t|\y^{1:t}) \, ,
\end{alignat}
where $\Expect(.)$ and $\Var(.)$ are the (multivariate)
expectation and variance operators.
In the linear-Gaussian case, these characterise $\pdf(\x^{t} \mid \y^{1:t-1})$
and $\pdf(\x^t \mid \y^{1:t})$,
and are given, recursively in time for sequentially increasing indices, $t$,
by the Kalman filter equations.

The EnKF is an algorithm to approximately sample ensembles, $\x_{1:N} = \{\x_n \mid \ninds\}$,
from these distributions.
Note that the positive integer $N$ is used to denote ensemble size,
while $m$ and $p$ have been used to denote state and observation vector lengths.
For convenience, all of the state realisations are assembled into the ``ensemble matrix'':
\begin{align}
	\label{eqn:Ea_3}
	\E &= \begin{bmatrix}
		\x_1, & \ldots & \x_n, & \ldots & \x_N
	\end{bmatrix} \, .
\end{align}
A related matrix is that of the ``anomalies'':
\begin{align}
	\A &= \E (\I_N - \bPi_{\ones}) = \E \AN \, ,
\end{align}
where $\ones \in \Reals^{N}$ is the column vector of ones, $\ones\tr$ is its transpose,
and the matrix $\I_N$ is the $N$-by-$N$ identity.
The conventional estimators serve as ensemble counterparts
to the exact first and second order moments of \cref{eqn:KF_notation_x,eqn:KF_notation_P},
\begin{alignat}{2}
	\label{eqn:EnKF_notation_f}
	\bx^f &= \frac{1}{N} \E^f \ones \, ,  &\qquad   \barP^f &= \frac{1}{N-1} \A^f {\A^f}\tr \, , \\
	\label{eqn:EnKF_notation_a}
	\bx^a &= \frac{1}{N} \E^a \ones \, ,  &\qquad   \barP^a &= \frac{1}{N-1} \A^a {\A^a}\tr \, ,
\end{alignat}
where, again, the superscripts indicate the conditioning.
Furthermore, $\A$ (without any superscript)
is henceforth used to refer to the anomalies at an
intermediate stage in the forecast step, before model noise incorporation.
In summary, the superscript usage of the EnKF cycle is illustrated by
\newcommand{\fsteps}{
	\A^a
	\xrightarrow[\text{\cref{eqn:A_pure}}]{\text{Model integration,}}
	\A
	\xrightarrow[\text{incorporation}]{\text{Model noise}}
}
\newcommand{\asteps}{
	\A^f
	\xrightarrow[\text{eqns. (\ref{eqn:xa_2},\ref{eqn:A_analysis_4})}]{\text{Analysis}}
	\A^a
}
\begin{align*}
	\rlap{$
		\overbrace{
			\phantom{\fsteps \A^f_{t}}}^\text{Forecast step}
		$}
		\fsteps
		\underbrace{
			\asteps
		}_{\text{Analysis step}}
\end{align*}
Although the first $\A^a$ of the diagram is associated with the time step before
that of $\A$, $\A^f$, and the latter $\A^a$, this
ambiguity becomes moot by focusing on the analysis step and the forecast step separately.

\subsection{Layout}
\label{sec:Layout}
The proposed methods to account for model noise builds on the square root method of the analysis step,
which is described in \cref{sec:analysis_sqrt}.
The core of the proposed methods is then set forth in \cref{sec:forecast_sqrt}.
Properties of both methods are analysed in \cref{sec:Properties}. 
Alternative techniques, against which the proposed method is compared,
are outlined in \cref{sec:Alternative_approaches}.
Based on these alternatives, \cref{sec:res_noise} introduces methods
to account for the residual noise resulting from the core method.
It therefore connects to, and completes, \cref{sec:forecast_sqrt}.
The set-up and results of numerical experiments are given in
\cref{sec:Setup} and \cref{sec:Results}.
A summary is provided, along with final discussions, in \cref{sec:Discussion}.
The appendices provide additional details on the properties
of the proposed square root methods.

\section{The square root method in the analysis step}
\label{sec:analysis_sqrt}
Before introducing the square root method for the EnKF forecast step, which accounts for model noise,
we here briefly discuss the square root method in the analysis step.

\subsection{Motivation}
\label{sec:Sqrt_analysis_motivation}

It is desirable that
$\barP^{a/f} = \bP^{a/f}$ and $\bx^{a/f} = \x^{a/f}$
throughout the DA process.
This means that the Kalman filter equations,
with the ensemble estimates swapped in,
\begin{align}
	\label{eqn:K_1}
	\bK &= \barP^f \bH\tr (\bH \barP^f \bH\tr +\R)^{-1} \, ,
	\\
	\label{eqn:xa_1}
	\bx^a &= \bx^f + \bK [\y - \bH \bx^f] \, ,
	\\
	\label{eqn:Pa_1}
	\barP^a &= [\I_m-\bK \bH] \barP^f \, ,
\end{align}
should be satisfied by $\E^a$ from the analysis update.

Let $\Dobs \in \Reals^{p \times N}$ be a matrix whose columns
are drawn independently from $\NormDist(0,\R)$.
Unfortunately, the perturbed observations analysis update \citep{burgers1998analysis},
\begin{align}
	\label{eqn:PertObsEnKF}
	\E^a &= \E^f + \bK \left\{\y\ones\tr + \Dobs - \bH \E^f\right\} \, ,
\end{align}
only yields the intended covariance, \cref{eqn:Pa_1}, on average:
\begin{align}
	\Expect (\barP^a) &= [ \I_m-\bK \bH ] \barP^f \, ,
	\label{eqn:Pa_8}
\end{align}
where the expectation, $\Expect$, is taken with respect to $\Dobs$.

\subsection{Method}
\label{sec:Analysis_SQRT_Method}

On the other hand, the square root analysis update satisfies \cref{eqn:Pa_1} exactly.
Originally introduced to the EnKF by \citet{bishop2001adaptive},
the square root analysis approach was soon connected to
classic square root Kalman filters \citep{tippett2003ensemble}.
But while the primary intention of classic square root Kalman filters
was to improve on the numerical stability of the Kalman filter \citep{anderson1979},
the main purpose of the square root EnKF was rather to eliminate
the stochasticity and the accompanying sampling errors of the perturbed-observations analysis update
(\ref{eqn:PertObsEnKF}).

Assume that $p \leq m$, or that $\R$ is diagonal, or that $\R^{-1/2}$ is readily computed.
Then, both for notational and computational \citep{hunt2007efficient} simplicity,
let
\begin{alignat}{2}
	\label{eqn:def_S1}
	\bS &= \R^{-1/2} (\bH \A^f) /\sqrt{N-1} & &\in \Reals^{p \times N} \, ,
	\\
	\label{eqn:def_s2}
	\bss &= \R^{-1/2} [\y - \bH \bx^f] /\sqrt{N-1} & &\in \Reals^{p} \, ,
\end{alignat}
denote the ``normalised''
anomalies and mean innovation of the ensemble of observations.
Recalling \cref{eqn:EnKF_notation_a}
it can then be shown that \cref{eqn:K_1,eqn:xa_1,eqn:Pa_1} are satisfied if:
\begin{align}
	\label{eqn:xa_2}
	\bx^a &= \bx^f + \A^f \G^a \bS\tr \bss \, ,
	\\
	\label{eqn:Pa_2}
	\A^a {\A^a}\tr
	&=
	\A^f \G^a {\A^f}\tr \, ,
\end{align}
where the two forms of $\G^a$,
\begin{align}
	\label{eqn:Pa_5}
	\G^a
	&=
	\I_N - \bS\tr (\bS \bS\tr + \I_p )^{-1} \bS
	\\
	\label{eqn:Pa_6}
	&=
	(\bS\tr \bS + \I_N)^{-1} \, ,
\end{align}
are linked through the Woodbury identity \citep[e.g.][]{wunsch2006discrete}.
Therefore, if $\A^a$ is computed by
\begin{align}
	\A^a = \A^f \T^a \, ,
	\label{eqn:A_analysis_4}
\end{align}
with $\T^a$ being a matrix square root of $\G^a$,
then $\A^a$ satisfies \cref{eqn:Pa_1} exactly.
Moreover, ``square root update'' is henceforth the term used to refer to
any update of the anomalies through the right-multiplication of a transform matrix,
as in \cref{eqn:A_analysis_4}.
The ensemble is obtained by recombining
the anomalies and the mean:
\begin{align}
       \E^a &= \bx^a \ones\tr + \A^a \, .
       \label{eqn:recombine_x_A}
\end{align}

\subsection{The symmetric square root}
\label{sec:Choice_of_square_root}

\Cref{eqn:Pa_6} implies that $\G^a$ is SPD.
The matrix $\T^a$ is a square root of $\G^a$ if it satisfies
\begin{align}
	\G^a = \T^a { \T^a }\tr \, .
	\label{eqn:T_a_0}
\end{align}
However, by substitution into \cref{eqn:T_a_0} it is clear that
$\T^a \Omeg$ is also a square root of $\G^a$,
for any orthogonal matrix $\Omeg$.
There are therefore infinitely many square roots.
Nevertheless, some have properties that make them unique.
For example, the Cholesky factor is unique as the only triangular square root 
with positive diagonal entries.

Here, however, the square root of most interest
is the symmetric one, $\T^a_s = \V \Sig^{1/2} \V\tr$.
Here, $\V \Sig \V\tr = \G^a $ is an eigendecomposition of $\G^a$,
and $\Sig^{1/2}$ is defined as the entry-wise
\emph{positive} square root of $\Sig$ \citep[][Th. 7.2.6]{horn2012matrix}.
Its existence follows from the spectral theorem,
and its uniqueness from that of the eigendecomposition.
Note its distinction by the $s$ subscript.

It has been gradually discovered that the symmetric square root choice
has several advantageous properties for its use in \cref{eqn:A_analysis_4},
one of which is that the it does not affect the ensemble mean \citep[e.g.][]{wang2003comparison,evensen2009ensemble},
which is updated by \cref{eqn:xa_2} apart from the anomalies.
Further advantages are surveyed in \cref{sec:Properties},
providing strong justification for choosing the symmetric square root,
and strong motivation to extend the square root approach to the forecast step.

\section{The square root method in the forecast step}
\label{sec:forecast_sqrt}
\Cref{sec:analysis_sqrt} reviewed the square root update method for the analysis step of the EnKF.
In view of its improvements over the Monte-Carlo method,
it is expected that a similar scheme for incorporating the model noise into the forecast ensemble, $\E^f$,
would be beneficial.
\Cref{sec:SR_nores} derives such a scheme: \SRN{}.
First, however, \cref{sec:Forecast_sampling_errors} illuminates the motivation:
forecast step sampling error.

\subsection{Forecast sampling errors in the classic EnKF}
\label{sec:Forecast_sampling_errors}
Assume linear dynamics, $f: \x \mapsto f(\x) = \F \x$, for ease of illustration.
The Monte-Carlo simulation of \cref{eqn:state_sde_discrete} can be written
\begin{align}
	\label{eqn:ens_evo}
	\E^f &= \F \E^a + \D \, ,
\end{align}
where the columns of $\D$ are drawn from $\NormDist(0,\Q)$ by
\begin{align}
	\label{eqn:add_infl_D}
	\D &= \Q^{1/2} \bXi \, ,
\end{align}
where $\bXi = \begin{bmatrix}
       \bxi_1, & \ldots & \bxi_n, & \ldots  & \bxi_N
\end{bmatrix}$,
and each $\bxi_n$ is independently drawn from $\NormDist(0,\I_m)$.
Note that different choices of the square root, say $\Q^{1/2}$ and $\Q^{1/2} \Omeg$,
yield equally-distributed random variables, $\Q^{1/2}\bxi$ and $\Q^{1/2} \Omeg \bxi$.
Therefore the choice does not matter, and is left unspecified.
It is typical to eliminate sampling error of the first order
by centering the model noise perturbations so that $\D \ones = 0$.
This introduces dependence between the samples and reduces the variance.
The latter is compensated for by rescaling by a factor of $\sqrt{N/(N-1)}$.
The result is that
\begin{align}
	\barP^f
	\label{eqn:bP3}
	&= \F \barP^a \F\tr + \Q \\
	 + &(\barQ - \Q)
	 - \frac{1}{N-1} \left( \F \A^a \D\tr
	+ \D (\F \A^a)\tr \right) \, , \notag
\end{align}
as per \cref{eqn:EnKF_notation_f},
where $\barQ = (N-1)^{-1} \D \D\tr$.
But, for the same reasons as for the analysis step, ideally:
\begin{align}
	\label{eqn:Pf1}
	\barP^f &= \F \barP^a \F\tr + \Q \, .
\end{align}
Thus, the second line of \cref{eqn:bP3} constitutes a stochastic discrepancy
from the desired relations (\ref{eqn:Pf1}).

\subsection{The square root method for model noise -- \SRNb{}}
\label{sec:SR_nores}

As illustrated in \cref{sec:Framework},
define $\A$ as the anomalies of the propagated ensemble before noise incorporation:
\begin{align}
	\label{eqn:A_pure}
	\A &= f(\E^a) \AN \, ,
\end{align}
where $f$ is applied column-wise to $\E^a$.
Then the desired relation (\ref{eqn:Pf1}) is satisfied if $\A^f$ satisfies:
\begin{align}
	\label{eqn:AfAf_PQ}
	\A^f{\A^f}\tr
	&= \A \A\tr + (N-1) \Q \, .
\end{align}
However, $\A^f$ can only have $N$ columns. Thus, the problem of finding an $\A^f$ that satisfies \cref{eqn:AfAf_PQ}
is ill-posed, since the right hand side of \cref{eqn:AfAf_PQ} is of rank $m$ for arbitrary, full-rank $\Q$,
while the left hand side is of rank $N$ or less.

Therefore, let $\A\pinv$ be the Moore-Penrose pseudoinverse of $\A$,
denote $\PiA = \A \A\pinv$ the orthogonal projector onto the column space of $\A$,
and define $\hQ = \PiA \Q \PiA$ the ``two-sided'' projection of $\Q$.
Note that the orthogonality of the projector, $\PiA$, induces its symmetry.
Instead of \cref{eqn:AfAf_PQ}, the core square root model noise incorporation method
proposed here, \SRN{}, only aims to satisfy
\begin{align}
	\label{eqn:AQhat_cond}
	\A^f{\A^f}\tr
	&= \A \A\tr +  (N-1) \hQ \, .
\end{align}
By virtue of the projection, \cref{eqn:AQhat_cond} can be written as
\begin{align}
	\label{eqn:T4}
	\G^f &= \I_N  +  (N-1) \A\pinv \Q (\A\pinv)\tr \, , \\
	\label{eqn:AGA_f}
	\A^f{\A^f}\tr
	&= \A \G^f \A\tr \, .
\end{align}
Thus, with $\T^f$ being a square root of $\G^f$, the update
\begin{align}
	\A^f = \A \T^f
	\label{eqn:f_sqrt_step}
\end{align}
accounts for the component of the noise quantified by $\hQ$.
The difference between the right hand sides of \cref{eqn:AfAf_PQ,eqn:AQhat_cond},
$(\compactN) \QQh$, is henceforth referred to as the ``residual noise'' covariance matrix.
Accounting for it is not trivial. This discussion is resumed in \cref{sec:res_noise}.

As for the analysis step, we choose to use the symmetric square root, $\T^f_s$, of $\G^f$.
Note that \emph{two} SVDs are required to perform this step: one to calculate $\A\pinv$,
and one to calculate the symmetric square root of $\G^f$. Fortunately, both are
relatively computationally inexpensive,
needing only to calculate $\compactN$ singular values and vectors.
For later use, define the square root ``additive equivalent'':
\begin{align}
	\hD &= \A^f - \A = \A [\T^f_s - \I_N] \, .
	\label{eqn:Dh_def}
\end{align}

\subsection{Preservation of the mean}
\label{sec:Preservation_of_the_mean}
The square root update is a deterministic scheme
that satisfies the covariance update relations exactly (in the space of $\A$).
But in updating the anomalies, the mean should remain the same.
For \SRN{}, this can be shown to hold true in the same way as
\citet{livings2008unbiased} did for the analysis step,
with the addition of \cref{eqn:Apinv_12}.
\begin{theo}[Mean preservation] If $\A^f = \A \T^f_s$, then
	\label{theo:pres_mean}
	\begin{align}
		\A^f \ones = 0 \, .
		\label{eqn:Af_ones_0}
	\end{align}
	I.e. the symmetric square root choice for the model noise transform matrix preserves the ensemble mean.
\end{theo}
\begin{proof}
	For any matrix $\A$,
	\begin{align}
		\A\pinv &= \A\tr (\A \A\tr)\pinv \, ,
		\label{eqn:Apinv_12}
	\end{align}
	\citep[][\S 1.6]{ben2003generalized}.
	Thus, 
	\begin{align}
		\G^f \ones
		\label{eqn:pres_mean_3}
		&= \ones + (N-1) \A\pinv \Q (\A \A\tr)\pinv \A \ones
		= \ones \, ,
	\end{align}
	as per \cref{eqn:A_pure}.
	But the eigenvectors of the square of a diagonalisable matrix are the same as for the original matrix,
	with squared eigenvalues.
	Thus \cref{eqn:pres_mean_3} implies $\A^f \ones = \A \T^f_s \ones = \A \ones = 0$.
\end{proof}

\section{Dynamical consistency of square root updates}
\label{sec:Properties}
Many dynamical systems embody ``balances'' or constraints on the state space \citep{van2009particle}.
For reasons of complexity and efficiency 
these concerns are often not encoded in the prior \citep{wang2015alleviating}.
They are therefore not considered by the statistical updates,
resulting in state realisations that are inadmissible because of a
lack of dynamical consistency or physical feasibility.
Typical consequence of breaking such constraints include unbounded growth (``blow up''), 
exemplified by the quasi-geostrophic model of \citet{sakov2008deterministic},
or failure of the model to converge, exemplified by reservoir simulators \citep{chen2013levenberg}.

This section provides a formal review of the properties of the square root update
as regards dynamical consistency, presenting theoretical support for the square root method.
The discussion concerns any square root update,
and is therefore relevant for the square root method in the analysis step
as well as for \SRN{}.

\subsection{Affine subspace confinement}
\label{sec:Linearity}
The fact that the square root update $\A \mapsto \A \T$ is a right-multiplication means that
each column of the updated anomalies is a linear combination of the original anomalies.
On the other hand, $\T$ itself depends on $\A$.
In recognition of these two aspects, \citet{evensen2003ensemble} called
such an update a ``weakly nonlinear combination''.
However, our preference is to describe the update as
confined to the affine subspace of the original ensemble,
that is the affine space $\bx + \vspan(\A)$.

\subsection{Satisfying equality constraints}
\label{sec:Constraints}
It seems reasonable to assume that the updated ensemble, being in the space of the original one,
stands a fair chance of being dynamically consistent.
However, if consistency can be described as equality constraints,
then discussions thereof can be made much more formal and specific,
as is the purpose of this subsection.
In so doing, it uncovers a couple of interesting, hitherto unnoticed advantage of the symmetric square root choice.

Suppose the original ensemble, $\x_{1:N}$, or $\E$, satisfies $\C \x_n = \d$ for all $\ninds$, i.e.
\begin{align}
	\label{eqn:lin_constraints}
	\C \E &= \d \ones\tr \, .
\end{align}
One example is conservation of mass,
in which case the state, $\x$, would contain grid-block densities,
while the constraint coefficients, $\C$,
would be a row vector of the corresponding volumes,
and $\d$ would be the total mass.
Another example is geostrophic balance \citep[e.g.][]{hoang2005adaptive},
in which case $\x$ would hold horizontal velocity components and sea surface heights,
while $\C$ would concatenate the identity and a discretised horizontal differentiation operator,
and $\d$ would be zero.

The constraints (\ref{eqn:lin_constraints}) should hold also after the update.
Visibly, if $\d$ is zero, any right-multiplication of $\E$,
i.e. any combination of its columns, will also satisfy the constraints.
This provides formal justification for the proposition of \citet{evensen2003ensemble},
that the ``linearity'' of the EnKF update implicitly ensures respecting linear constraints.

One can also write
\begin{align}
	\label{eqn:lin_constraint_mean_3}
	\C \bx &= \d \, , \\
	\label{eqn:lin_constraints_7}
	\C \A  &= 0 \ones\tr \, ,
\end{align}
implying (\ref{eqn:lin_constraints}) provided
$\E = \bx \ones\tr + \A$ holds.
\Cref{eqn:lin_constraint_mean_3,eqn:lin_constraints_7} show that
the ensemble mean and anomalies can be thought of as particular and homogeneous solutions to the constraints.
They also indicate that in a square root update, even if $\d$ is not zero,
one only needs to ensure that the mean constraints are satisfied,
because the homogeneity of \cref{eqn:lin_constraints_7} means that
any right-multiplying update to $\A$ will satisfy the anomaly constraints.
However, as mentioned above, unless it preserves the mean, it might perturb \cref{eqn:lin_constraint_mean_3}.
A corollary of \cref{theo:pres_mean} is therefore that the symmetric choice for the square root update
also satisfies inhomogeneous constraints.

Finally, in the case of nonlinear constraints, e.g. $\mathscr{C}(\x_n) = \d$,
truncating the Taylor expansion of $\mathscr{C}$ yields
\begin{align}
	\C \A \approx [\d - \mathscr{C}(\bx)]\ones\tr \, ,
	\label{eqn:nonlin-constraint_8}
\end{align}
where $\C = \pd{\mathscr{C}}{\x}(\bx)$.
Contrary to \cref{eqn:lin_constraints_7}, the approximate constraints of \cref{eqn:nonlin-constraint_8},
are not homogeneous, and therefore not satisfied by any right-multiplying update.
Again, however, by \cref{theo:pres_mean}, the symmetric square root appears
an advantageous choice, because it has $\ones$ as an eigenvector with eigenvalue 1,
and therefore satisfies the (approximate) constraints.

\subsection{Optimality of the symmetric choice}
\label{sec:Optimal_transport}

A number of related properties on the optimality of the symmetric square root exist scattered in the literature.
However, to the best of our knowledge, these have yet to be reunited into a unified discussion.
Similarly, considerations on their implications on DA have so far not been collected.
These are the aims of this subsection.

\begin{theo}[Minimal ensemble displacement]
	\label{theo:T_sym_opt}
	Consider the ensemble anomalies $\A$ with ensemble covariance matrix $\barP$,
	and let $\q_n$ be column $n$ of
	$\D = \A \T - \A$: the displacement of the $n$-th anomaly through a square root update.
	The symmetric square root, $\T_s$, minimises
	\begin{align}
		J(\T)
		\label{eqn:J_ott_1}
		&= \frac{1}{N-1} \sum_{n}^{} \norm{\q_n}_{\barP}^2 \\
		&= \trace \left( \left[ \A \T - \A \right]\tr (\A \A\tr)\pinv \left[ \A \T - \A \right] \right)
		\label{eqn:J_ott_3}
	\end{align}
	among all $\T \in \Reals^{N^2}$ such that $\A \T \T\tr \A\tr = \A \G \A\tr$, for some SPD matrix $\G$.
	\Cref{eqn:J_ott_3} coincides with \cref{eqn:J_ott_1} if $\barP^{-1}$ exists,
	but is also valid if not.
\end{theo}
\cref{theo:T_sym_opt} was proven by \citet{ott2004local},
and later restated by \citet{hunt2007efficient} as the constrained optimum
of the Frobenius norm of $\left[\T - \I_N \right]$.
Another interesting and desirable property of the symmetric square root is the fact that the updated ensemble
members are all equally likely realisations of the estimated posterior \citep{wang2004better,mclay2008evaluation}.
More recently, the choice of mapping between the original and the updated ensembles
has been formulated through optimal transport theory
\citep{cotter2012ensemble,oliver2014minimization}.
However, the cost functions therein typically use a different weighting on the norm than $J(\T)$,
in one case yielding an optimum that is the symmetric \emph{left}-multiplying transform matrix
-- not to be confused with the right-multiplying one of \cref{theo:T_sym_opt}.

\Cref{theo:T_sym_opt} and the related properties should benefit
the performance of filters employing the square root update,
whether for the analysis step, the model noise incorporation, or both.
In part, this is conjectured because minimising the displacement of an update
means that the ensemble cloud should retain some of its shape,
and with it higher-order, non-Gaussian information, as illustrated in \cref{fig:snapshot}.

A different set of reasons to expect strong performance from the symmetric square root choice
is that it should promote dynamical consistency,
particularly regarding inequality constraints, such as the inherent positivity of concentration variables,
as well as non-linear equality constraints, initially discussed in \cref{sec:Constraints}.
In either case it stands to reason that smaller displacements are less likely to break the constraints,
and therefore that their minimisation should inhibit it.
Additionally, it is important when using ``local analysis'' localisation
that the ensemble is updated similarly at nearby grid points.
Statistically, this is ensured by employing smoothly decaying localisation functions,
so that $\G$ does not jump too much from one grid point to the next.
But, as pointed out by \citet{hunt2007efficient},
in order to translate this smoothness to dynamical consistency,
it is also crucial that the square root is continuous in $\G$.
Furthermore, even if $\G$ does jump from one grid point to the next,
it still seems plausible that the minimisation of displacement might
restrain the creation of dynamical inconsistencies.

\section{Alternative approaches}
\label{sec:Alternative_approaches}

\begin{table*}[h]
	\caption[Comparison of some model noise incorporation methods.]
	{Comparison of some model noise incorporation methods.
		\vspace{-3ex}}\label{t4}
\begin{center}
\begin{tabular}{l l l l l l}
\toprule
\textbf{Description} & \textbf{Label} & $\A^f = \qquad$ & \textbf{where} & & \textbf{thus satisfying}\\
\midrule
Additive, simulated noise & \ADD{} & $\A + \D$ & $\D$ is a centred sample from $\NormDist(0,\Q)$ 
& & $\Expect_{\D}(\text{\cref{eqn:AfAf_PQ}})$ \\
Scalar inflation & \MULS{} & $\lambda \A$ & $\lambda^2 = \trace( \barP )^{-1} \trace( \barP + \Q)$ 
& & $\trace(\text{\cref{eqn:AfAf_PQ}})$ \\
Multivariate inflation & \MULT{} & $\Lambd \A$ & $\Lambd^2 = \diag( \barP )^{-1} \diag( \barP + \Q)$
& & $\diag(\text{\cref{eqn:AfAf_PQ}})$ \\
Core square root method & \SRN{} & $\A \T$ & $\T = \big(\I_N + (\compactN)\A\pinv \Q \A\pinvtr\big)^{1/2}_s$
& & $\PiA(\text{\cref{eqn:AfAf_PQ}})\PiA$ \\
\bottomrule
\end{tabular}
\end{center}
\end{table*}
This section describes the model noise incorporation methods most relevant methods to this study.
\cref{t4} summarises the methods that will be used in numerical comparison experiments.
\ADD{} is the classic method detailed in \cref{sec:Forecast_sampling_errors}.
\MULS{} and \MULT{} are multiplicative inflation methods.
The rightmost column relates the different methods to each other by
succinctly expressing the degree to which they satisfy \cref{eqn:AfAf_PQ};
it can also be used as a starting point for their derivation.
Note that \MULS{} only satisfies one degree of freedom of \cref{eqn:AfAf_PQ},
while \MULT{} satisfies $m$ degrees, and would therefore be expected to perform better in general.
It is clear that \MULS{} and \MULT{} will generally not provide an exact statistical update,
no matter how big $N$ is,
while \ADD{} reproduces \emph{all} of the moments almost surely as $N \rightarrow \infty$.
By comparison, \SRN{} guarantees obtaining the correct first two moments for any $N>m$,
but does not guarantee higher order moments.

Using a large ensemble size, \cref{fig:snapshot} illustrates the different techniques.
Notably, the cloud of \ADD{} is clearly more dispersed than any of the other methods.
Furthermore, in comparison to \MULT{} and \MULS{}, \SRN{} significantly skewers
the distribution in order to satisfy the off-diagonal conditions.
\begin{figure}[ht]
	\centerline{\includegraphics[width=19pc]{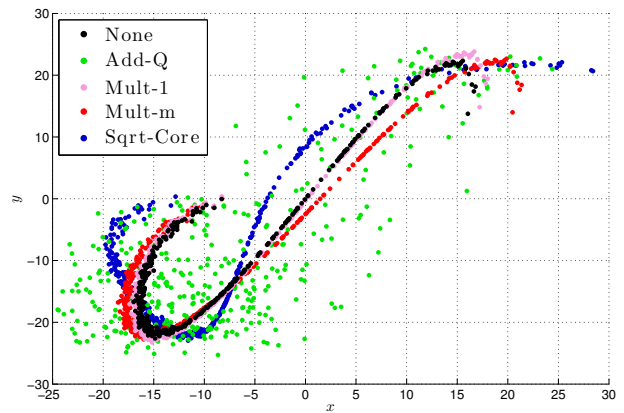}}
	\caption[Scatter plot of ensemble forecasts with the three-dimensional
		Lorenz-63 system \citep{lorenz1963deterministic}
		using different schemes to account for the model noise.]
	  {Scatter plot of ensemble forecasts with the three-dimensional
		Lorenz-63 system \citep{lorenz1963deterministic}
		using different schemes to account for the model noise, 
		which is specified by $\dt Q = \diag([36.00, 3.60, 1.08])$
		and makes up approximately $30\%$ of the total spread of the updated ensembles.
		Each dot corresponds to the ``$(x,y)$'' coordinate of one realisation among $N=400$.
	}
	\label{fig:snapshot}
\end{figure}

Continuing from \cref{sec:Survey},
the following details other pertinent alternatives,
some of them sharing some similarity with the square root methods proposed here.

One alternative is to resample the ensemble fully
from $\NormDist(0,\A\A\tr/(\compactN) + \Q)$.
However, this incurs larger sampling errors than \ADD{},
and is more likely to cause dynamical inconsistencies.

Second-order exact sampling \citep{pham2001stochastic} attempts to sample noise
under the restriction that all of the terms on the second line of \cref{eqn:Pf1} be zero.
It requires a very large ensemble size ($N > 2m$), and is therefore typically not applicable,
though recent work indicate that this might be circumvented
\citep{hoteit2015mitigating}.

The singular evolutive interpolated Kalman (SEIK) filter \citep{hoteit2002simplified}
has a slightly less primitive and intuitive formalism than the EnKF,
typically working with matrices of size $m\times (N-1)$.
Moreover, it does not have a separate step to deal with model noise,
treating it instead implicitly, as part of the analysis step.
This lack of modularity has the drawback that the frequency
of model noise incorporation is not controllable:
in case of multiple model integration steps between observations,
the noise should be incorporated at each step in order to evolve with the dynamics;
under different circumstances, skipping the treatment of noise for a few steps
can be cost efficient \citep{evensen1996assimilation}.
Nevertheless, a stand-alone model noise step can be distilled from the SEIK algorithm as a whole.
Its forecast covariance matrix, $\barP^f$, would equal to that of \SRN{}: $\PiA (\barP + \Q) \PiA$.
However, unlike \SRN{}, which uses the symmetric square root,
the SEIK uses random rotation matrices to update the ensemble.
Also, the SEIK filter uses a ``forgetting factor''.
Among other system errors, this is intended to account for the residual noise covariance, $\QQh$.
As outlined in \cref{sec:Survey},
however, this factor is not explicitly a function of $\QQh$;
it is instead obtained from manual tuning.
Moreover, it is only applied in the update of the ensemble mean.

Another method is to include only the $N-1$ largest eigenvalue components of $\barP + \Q$,
as in reduced-rank square root filters \citep{verlaan1997tidal},
and some versions of the unscented Kalman filter \citep{chandrasekar2008reduced}.
This method can be referred to as \TSVD{} because the update can be effectuated
through a truncated SVD of $[\barP^{1/2}, \Q^{1/2}]$, where the choices of square roots do not matter.
It captures more of the total variance than \SRN{}, but also changes the ensemble subspace.
Moreover, it is not clear how to choose the updated ensemble.
For example, one would suspect dynamical inconsistencies to arise from using the ordered sequence of the truncated SVD.
Right-multiplying by random rotation matrices, as in the SEIK, might be a good solution.
Or, if computed in terms of a \emph{left}-multiplying transform matrix,
the symmetric choice is likely a good one.
Building on \TSVD{}, the ``partially orthogonal'' EnKF and the COFFEE algorithm
of \citep{heemink2001variance,hanea2007hybrid} also recognise the issue of the residual noise.
In contrasts with the treatments proposed in this study,
these methods introduce a complementary ensemble to account for it.

\section{Improving \SRN{}: Accounting for the residual noise}
\label{sec:res_noise}
As explained in \cref{sec:Linearity},
\SRN{} can only incorporate noise components that are in the span (range) of $\A$.
This leaves a residual noise component unaccounted for,
orthogonal to the span of $\A$,
with $\QQh$ posing as its covariance matrix.

First consider \emph{why} there is no such residual of $\R$ for the square root methods in the analysis step:
because the analysis step subtracts uncertainty, unlike the forecast step which adds it.
Therefore the presence or absence of components of $\R$ outside of the span of the observation ensemble
makes no difference to the analysis covariance update because
the ensemble effectively already assumes zero uncertainty in these directions.

In the rest of this section the question addressed is how to deal with the residual noise.
It is assumed that \SRN{}, \cref{eqn:f_sqrt_step}, has already been performed.
The techniques proposed thus \emph{complement} \SRN{},
but do not themselves possess the beneficial properties of \SRN{} discussed in \cref{sec:Properties}.
Also, the notation of the previous section is reused.
Thus, the aim of this section is to find an $\A^f \in \Reals^{m \times N}$
that satisfies, in some limited sense
\begin{align}
	\A^f {\A^f}\tr &= \A {\A}\tr + (N-1) \QQh \, .
	\label{eqn:AfAf_PQhat}
\end{align}

\subsection{Complementary, additive sampling -- \SRAb{}}
\label{sec:add_res_Z}
Let $\Q^{1/2}$ be a \emph{any} square root of $\Q$, and define
\begin{align}
	\label{eqn:Qhh_def}
	\Qhh
	&= \PiA \Q^{1/2} \, , \\
	\label{eqn:Z_def}
	\Z &=
	(\I_m - \PiA) \Q^{1/2} \, ,
\end{align}
the orthogonal projection of $\Q^{1/2}$ onto the column space of $\A$,
and the complement, respectively.

A first suggestion to account for the residual noise is
to use one of the techniques of \cref{sec:Alternative_approaches},
with $\QQh$ taking the place of the full $\Q$ in their formulae.
In particular, with \ADD{} in mind, the fact that
\begin{align}
	\Q^{1/2} = \Qhh + \Z
	\label{eqn:Q_is_Q_p_Z}
\end{align}
motivates sampling the residual noise using $\Z$.
That is, in addition to $\hD$ of \SRN{}, which accounts for $\hQ$,
one also adds $\tD = \Z \tXi$ to the ensemble, where the columns of $\tXi$
are drawn independently from $\NormDist(0,\I_m)$.
We call this technique \SRA{}.

Note that $\Qhh$, defined by \cref{eqn:Qhh_def}, is a square root of $\hQ$.
By contrast, multiplying \cref{eqn:Q_is_Q_p_Z} with its own transpose yields
\begin{align}
	\Z \Z\tr
	&= \QQh - \Qhh \Z\tr - \Z \Qhtr \, ,
	\label{eqn:zz1}
\end{align}
and reveals that $\Z$ is not a square root of $\QQh$.
Therefore, with expectation over $\tXi$,
\SRA{} does not respect $\Expect(\text{\cref{eqn:AfAf_PQhat}})$,
as one would hope.

Thus, \SRA{} has a bias equal to the cross term sum, $\Qhh \Z\tr + \Z \Qhtr = \QQh - \Z \Z\tr$.
Notwithstanding this problem, \cref{coro:trace_zero} of appendix A
shows that the cross terms sum,
has a spectrum symmetric around 0, and thus zero trace.
To some extent, this exonerates \SRA{}, since it means that the expected total variance is unbiased.

\subsection{The underlying problem: replacing a single draw with two independent draws}
\label{sec:repl_one_with_two}

Since any element of $\hQ$ is smaller than the corresponding element in $\Q$,
either one of the multiplicative inflation techniques
can be applied to account for $\QQh$ without second thoughts.
Using \MULS{} would satisfy $\trace (\text{\cref{eqn:AfAf_PQhat}})$,
while \MULT{} would satisfy $\diag(\text{\cref{eqn:AfAf_PQhat}})$.
However, the problem highlighted for \SRA{} is not just a technicality.
In fact, as shown in appendix A section \ref{sec:the_res},
$\QQh$ has negative eigenvalues because of the cross terms.
It is therefore not a valid covariance matrix in the sense that it has no real square root:
samples with covariance $\QQh$ will necessarily be complex numbers;
this would generally be physically unrealisable and therefore inadmissible.
This underlying problem seems to question the validity of the whole approach
of splitting up $\Q$ and dealing with the parts $\hQ$ and $\QQh$ separately.

Let use emphasise the word independently, because that is, to a first approximation,
what we are attempting to do:
replacing a single draw from $\Q$ by one from $\hQ$ plus another, independent draw from $\QQh$.
Rather than considering $N$ anomalies, let us now focus on a single one, and drop the $n$ index.
Define the two random variables,
\begin{align}
	\label{eqn:qpar}
	\qpar &= \Qhh \xipar + \Z \xipar \, , \\
	\label{eqn:qperp}
	\qperp &= \Qhh \hxi + \Z \txi \, ,
\end{align}
where $\xipar, \hxi, \txi$ are random variables independently drawn from
$\NormDist(0,\I_m)$.
By \cref{eqn:Q_is_Q_p_Z}, and design,
$\q$ can be identified with any of the columns of $\D$ of \cref{eqn:add_infl_D}
and, furthermore, $\Var(\q) = \Q$.
On the other hand, while $\qpar$ originates in a single random draw,
$\qperp$ is the sum of two independent draws.

The dependence between the terms of $\qpar$, and the lack thereof for $\qperp$, yields
the following discrepancy between the variances:
\begin{align}
	\label{eqn:var_qpar}
	\Var(\qpar) &= \hQ + \Z \Z\tr + \Qhh \Z\tr + \Z \Qhtr \, ,
	\\
	\label{eqn:var_qperp}
	\Var(\qperp) &= \hQ + \Z \Z\tr \, .
\end{align}
Formally, this is the same problem that was identified with \cref{eqn:zz1},
namely that of finding a real square root of $\QQh$, or eliminating the cross terms.
But \cref{eqn:var_qpar,eqn:var_qperp} show that the problem arises
from the more primal problem of trying to emulate $\qpar$ by $\qperp$.
Vice versa, $\Qhh \Z\tr = 0$ would imply that the ostentatiously dependent terms, $\Qhh \bxi$ and $\Z \bxi$,
are independent, and thus $\qperp$ is emulated by $\qpar$.

\subsection{Reintroducing dependence -- \SRDb{}}
\label{sec:reintro_dep}
As already noted, though, making the cross terms zero is not possible for general $\A$ and $\Q$.
However, the perspective of $\qpar$ and $\qperp$ hints at another approach:
reintroducing dependence between the draws.
In this section we will reintroduce dependence by making the residual sampling
depend on the square root equivalent, $\hD$ of \cref{eqn:Dh_def}.

The trouble with the cross terms
is that $\Q$ ``gets in the way'' between $\PiA$ and $(\I_m - \PiA)$,
whose product would otherwise be zero.
Although less ambitious than emulating $\qpar$ with $\qperp$,
it is possible to emulate a single draw from $\NormDist[0,\I_m]$, e.g. $\xipar$,
with two independent draws:
\begin{align}
	\xiperp &= \bPi \hxi + (\I_m-\bPi) \txi \, ,
	\label{eqn:XiPi}
\end{align}
where, as before, $\hxi$ and $\txi$ are independent random variables with law $\NormDist(0,\I_m)$,
and $\bPi$ is some orthogonal projection matrix.
Then, as the cross terms cancel,
\begin{align}
	\label{eqn:Var_xiperp_1}
	\bPi \bPi\tr + (\I_m-\bPi)(\I_m-\bPi)\tr = \I_m \, ,
\end{align}
and thus $\Var(\xiperp) = \Var(\xipar)$.

We can take advantage of this emulation possibility by choosing $\bPi$
as the orthogonal projector onto the \emph{rows} of $\Qhh$.
Instead of \cref{eqn:qpar}, redefine $\q$ as
\begin{align}
	\label{eqn:q_fin_def}
	\q
	&=
	\Q^{1/2} \xiperp \, .
\end{align}
Then, since $\Var(\xiperp) = \I_m$,
\begin{align}
	\Var(\q) &= \Q^{1/2} \I_m \Q\trsqrt = \Q \, ,
	\label{eqn:Var_q_fin}
\end{align}
as desired. But also
\begin{align}
	\label{eqn:Var_q_fin_3}
	\q
	&= (\Qhh + \Z) \left( \PiQ \hxi  + (\I_m-\PiQ) \txi \right)
	\\
	\label{eqn:Var_q_fin_7}
	&= \Qhh \hxi + \Z \left( \PiQ \hxi  + (\I_m-\PiQ) \txi \right) \, .
\end{align}
The point is that, while maintaining $\Var(\q) = \Q$,
and despite the reintroduction of dependence between the two terms in \cref{eqn:Var_q_fin_7},
the influence of $\txi$ has been confined to $\vspan(\Z) = \vspan(\A)^\perp$.
The above reflections yield the following algorithm, labelled \SRD{}:
\begin{compactenum}
	\item Perform the core square root update for $\hQ$, \cref{eqn:f_sqrt_step};
	\item Find $\hXi$ such that $ \Qhh_s \hXi = \hD$ of \cref{eqn:Dh_def}.
		Components in the kernel of $\Qhh_s$ are inconsequential;
	\item Sample $\tXi$ by drawing each column independently from $\NormDist(0,\I_m)$;
	\item Compute the residual noise, $\tD$, and add it to the ensemble anomalies;
		\begin{align}
			\tD = \Z \left( \PiQ \hXi  + (\I_m-\PiQ) \tXi \right) \, .
			\label{eqn:D_tilde_7}
		\end{align}
\end{compactenum}
Unfortunately, this algorithm requires the additional SVD of $\Qhh$ in order to compute 
$\PiQ$ and $\hXi$. Also, despite the reintroduction of dependence,
\SRD{} is not fully consistent, as discussed in appendix B.

\section{Experimental set-up}
\label{sec:Setup}
The model noise incorporation methods detailed in \cref{sec:forecast_sqrt,sec:res_noise}
are benchmarked using ``twin experiments'',
where a ``truth'' trajectory is generated and subsequently estimated by the ensemble DA systems.
As indicated by \cref{eqn:state_sde_discrete,eqn:obs_eqn},
stochastic noise is added to the truth trajectory and observations, respectively.
As defined in \cref{eqn:state_sde_discrete},
$\Q$ implicitly includes a scaling by the model time step,
$\dt$, which is the duration between successive time indices.
Observations are not taken at every time index,
but after a duration, $\dtObs$, called the DA window, which is a multiple of $\dt$.

The noise realisations excepted,
the observation process, \cref{eqn:obs_eqn}, given by $\bH$, $\R$, and $\dtObs$,
and the forecast process, \cref{eqn:state_sde_discrete}, given by $f$, $\bmu^0$, $\bP^0$ and $\Q$,
are both perfectly known to the DA system.
The analysis update is performed using the
symmetric square root update of \cref{sec:analysis_sqrt}
for all of the methods under comparison.
Thus, the only difference between the ensemble DA systems is their model noise incorporation method.

Performance is measured by the root-mean-square error of the ensemble mean,
given by:
\begin{align}
	\RMSE = \sqrt{\frac{1}{m} \norm{\bx^t - \x^t}_2^2} \, ,
	\label{eqn:RMSE_3}
\end{align}
for a particular time index $t$.
By convention, the RMSE is measured only immediately following each analysis update.
In any case, there was little qualitative difference to ``forecast'' RMSE averages,
which are measured right \emph{before} the analysis update.
The score is averaged for all analysis times after an
initial transitory period whose duration is estimated beforehand by studying the RMSE time series.
Each experiment is repeated 16 times with different initial random seeds.
The empirical variances of the RMSEs are checked to ensure satisfying convergence.

Covariance localisation is not used.
Following each analysis update, the ensemble anomalies are rescaled by a
scalar inflation factor intended to compensate for the consequences of sampling error in the analysis
\citep[e.g.][]{anderson1999monte,bocquet2011ensemble}.
This factor, listed in \cref{t1}, was approximately optimally tuned
prior to each experiment.
In this tuning process the \ADD{} method was used for the forecast noise incorporation,
putting it at a slight advantage relative to the other methods.

In addition to the EnKF with different model incorporation methods,
the twin experiments are also run with the standard methods of \cref{t4}
for comparison, as well as three further baselines:
(a) the climatology, estimated from several long, free runs of the system,
(b) 3D-Var (optimal interpolation) with the background from the climatology, and
(c) the extended Kalman filter \citep{rodgers2000inverse}.

\subsection{The linear advection model}
\label{sec:LA}

The linear advection model evolves according to
\begin{align}
	x^{t+1}_i = 0.98 x^{t}_{i-1}
	\label{eqn:LA_1}
\end{align}
for $t = 0,\ldots$, $\iinds$, with $m=1000$, and periodic boundary conditions.
The dissipative factor is there to counteract amplitude growth due to model noise.
Direct observations of the truth are taken at $p=40$ equidistant locations,
with $\R = 0.01 \I_p$, every fifth time step.

The initial ensemble members, $\{\x_n^0 \mid \ninds\}$, as well as the truth, $\x^0$,
are generated as a sum of 25 sinusoids of random amplitude and phase,
\begin{align}
	x_{i,n}^0 = \frac{1}{c_n} \sum_{k=1}^{\kMax} a_n^k \sin \left(2\pi k \left[ i/m + \varphi_n^k \right]\right) \, ,
	\label{eqn:LA_x0}
\end{align}
where $a_n^k$ and $\varphi_n^k$
is drawn independently and uniformly from the interval $(0,1)$
for each $n$ and $k$,
and the normalisation constant, $c_n$, is such that the standard deviation of each $\x_n^0$ is 1.
Note that the spatial mean of each realisation of \cref{eqn:LA_x0} is zero.
The model noise is given by
\begin{align}
	\Q &= 0.01\Var(\x^0) \, .
	\label{eqn:Q_LA}
\end{align}

\subsection{The Lorenz-96 model}
\label{sec:L40}
The Lorenz-96 model evolves according to
\begin{align}
	\dd{\x_i}{t}
	=
	\left( x_{i+1} - x_{i-2} \right) x_{i-1} - x_{i} + F \, ,
	\label{eqn:L96_1}
\end{align}
for $t>0$, and $\iinds$, with periodic boundary conditions.
It is a nonlinear, chaotic model that mimics the atmosphere at a certain latitude circle.
We use the parameter settings of \citet{lorenz1998optimal},
with a system size of $m = 40$, a forcing of $F = 8$, and the fourth-order Runge-Kutta numerical time stepping scheme
with a time step of $\dt = 0.05$.
Unless otherwise stated, direct observations of the entire state vector are taken
a duration of $\dtObs = 0.05$ apart, with $\R = \I_m$.

The model noise is spatially homogeneous, generated using a Gaussian autocovariance function,
\begin{align}
	\Q_{i,j} &= \exp\left(-1/30 \norm{i-j}_2^2\right) + 0.1\delta_{i,j} \, ,
	\label{eqn:L40_Gauss_Q}
\end{align}
where the Kronecker delta, $\delta_{i,j}$, has been added for numerical stability issues.

\section{Experimental results}
\label{sec:Results}

Each figure contains the results from a set of experiments run for a range of some control variable.

\subsection{Linear advection}
\label{sec:Linear_advection}

\Cref{fig:LA_m6} shows the RMSE versus the ensemble size for different
model noise incorporation schemes.
The maximum wavenumber of \cref{eqn:LA_x0} is $k=\kMax$.
Thus, by the design of $\bP^0$ and $\Q$, the dynamics will
take place in a subspace of rank 50, even though $m=1000$.
This is clearly reflected in the curves of the square root methods,
which all converge to the optimal performance of the Kalman filter (0.15)
as $N$ approaches $51$, and $\Z$ goes to zero.
\SRA{} takes a little longer to converge because of numerical error.
The multiplicative inflation curves are also constant for $N \geq 51$,
but they do not achieve the same level of performance.
As one would expect, \ADD{} also attains the performance of the Kalman filter as $N \rightarrow \infty$.
\begin{figure}
	\centerline{\includegraphics[width=19pc]{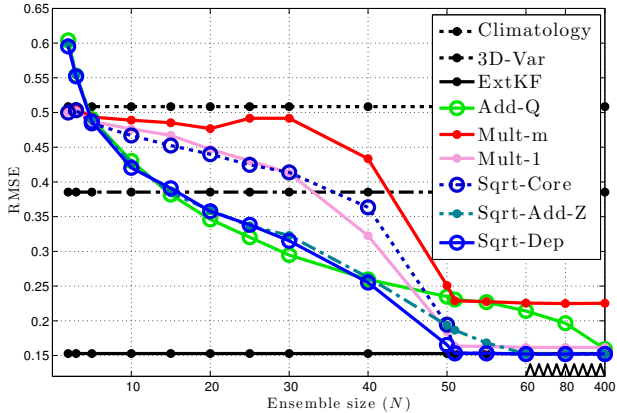}}
	\caption[
		Performance benchmarks as a function of the ensemble size, $N$,
		obtained with the linear advection system.
	]
	{
		Performance benchmarks as a function of the ensemble size, $N$,
		obtained with the linear advection system.
		The scale has been irregularly compressed for $N > 60$.
	}
	\label{fig:LA_m6}
\end{figure}

Interestingly, despite \MULT{} satisfying \cref{eqn:AfAf_PQ} to a higher degree
than \MULS{}, the latter performs distinctly better across the whole range of $N$.
This can likely be blamed on the fact that \MULT{} has the adverse effect of changing the subspace of the ensemble,
though it is unclear why its worst performance occurs near $N = 25$.

\ADD{} clearly outperforms \MULS{} in the intermediate range of $N$,
indicating that the loss of nuance in the covariance matrices of \MULS{}
is more harmful than the sampling error incurred by \ADD{}.
But, for $45 < N < 400$, \MULS{} beats \ADD{}. It is not clear why this reversal happens.

\SRN{} performs quite similar to \MULS{}.
In the intermediate range, it is clearly deficient compared to the square root methods that account for residual noise,
illustrating the importance of doing so.
The performance of \SRD{} is almost uniformly superior to all of the other methods.
The only exception is around $N=25$, where \ADD{} slightly outperforms it.
The computationally cheaper \SRA{} is beaten by \ADD{} for $N<40$,
but has a surprisingly robust performance nevertheless.

\subsection{Lorenz-96}
\label{sec:Lorenz96}

\begin{figure}
	\centerline{\includegraphics[width=19pc]{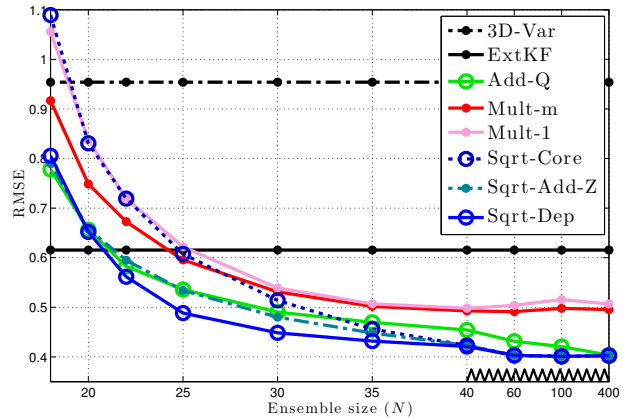}}
	\caption[
		Performance benchmarks as a function of the ensemble size, $N$,
		obtained with the Lorenz-96 system.
	]
	{Performance benchmarks as a function of the ensemble size, $N$,
		obtained with the Lorenz-96 system.
		The climatology averages an RMSE of 3.69 for both figures.
		The scale has been irregularly compressed for $N > 40$.
	}
	\label{fig:L40_N_m70}
\end{figure}

\Cref{fig:L40_N_m70} shows the RMSE versus ensemble size.
As with the linear advection model,
the curves of the square root schemes are coincident when $\Z = 0$,
which here happens for $N > m = 40$.
In contrast to the linear advection system, however,
the square root methods still improve as $N$ increases beyond $m$,
and noticeably so until $N = 60$.
This is because a larger enable is better able to characterise
the non-Gaussianity of the distributions and the non-linearity of the models.
On the other hand, the performance of the multiplicative inflation methods stagnates
around $N=m$, and even slightly deteriorates for larger $N$.
This can probably be attributed to the effects observed by \citet{sakov2008implications}.

Unlike the more ambiguous results of the linear advection model,
here \ADD{} uniformly beats the multiplicative inflation methods.
Again, the importance of accounting for the residual noise is highlighted
by the poor performance of \SRN{} for $N < 40$.
However, even though \SRA{} is biased, it outperforms \ADD{} for $N>25$,
and approximately equals it for smaller $N$.

The performance of \SRD{} is nearly uniformly the best,
the exception being at $N=18$, where it is marginally beaten by \ADD{} and \SRA{}.
The existence of this occurrence can probably be attributed to the slight suboptimality discussed in Appendix B,
as well as the advantage gained by \ADD{} from using it to tune the analysis inflation.
Note, though, that this region is hardly interesting, since results lie above the
baseline of the extended KF.

\ADD{} asymptotically attains the performance of the square root methods.
In fact, though it would have been imperceptible if added to \cref{fig:L40_N_m70},
experiments show that \ADD{} beats \SRD{} by an average RMSE difference of 0.005
at $N=800$, as predicted in \cref{sec:Alternative_approaches}.

\begin{figure}
	\centerline{\includegraphics[width=19pc]{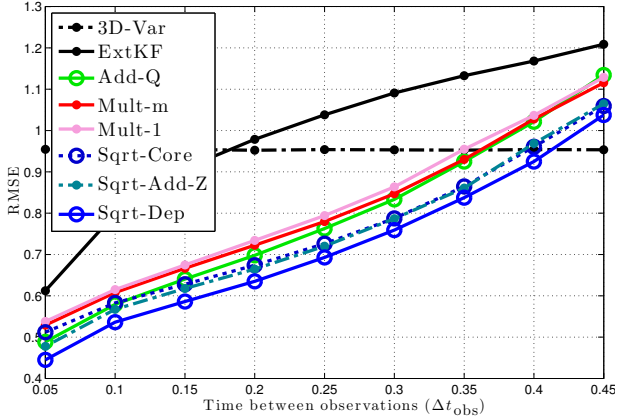}}
	\caption[
		Performance benchmarks as a function of the data assimilation window, $\dtObs$,
		obtained with the Lorenz-96 model.
	]
	{
		Performance benchmarks as a function of the data assimilation window, $\dtObs$,
		obtained with the Lorenz-96 model and $N=30$.
		The climatology averages an RMSE of 3.7.
	}
	\label{fig:L40_dtObs_m70}
\end{figure}
\Cref{fig:L40_dtObs_m70} shows the RMSE versus the DA window.
The performance of \ADD{} clearly deteriorates more than that of all
of the deterministic methods as $\dtObs$ increases.
Indeed, the curves of \SRN{} and \ADD{} cross at $\dtObs \approx 0.1$,
beyond which \SRN{} outperforms \ADD{}.
\SRN{} even gradually attains the performance of \SRA{},
though this happens in a regime where all of the EnKF methods are beaten by 3D-Var.
Again, however, \SRD{} is uniformly superior,
while \SRA{} is uniformly the second best.
Similar tendencies were observed in experiments (not shown) with $N=25$.

\begin{figure}
	\centerline{\includegraphics[width=19pc]{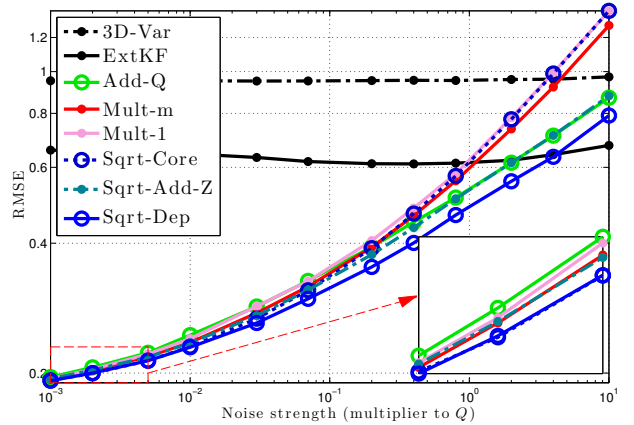}}
	\caption[
		Performance benchmarks as a function of the noise strength,
		obtained with the Lorenz-96 model.
	]
	{
		Performance benchmarks as a function of the noise strength,
		obtained with the Lorenz-96 model and $N=25$.
		Both axes are logarithmic.
		On average, when $\Q$ is multiplied by $10^{-3}$ (resp. $10^{-2},10^{-1},10^{0},10^{1}$),
		the model noise makes up approximately $0.5$ (resp. $4,20,70,90$) percent
		of the growth in the spread of the ensemble. 
		The climatology averages an RMSE score of approximately 4.
	}
	\label{fig:L40_Q_m70}
\end{figure}
\Cref{fig:L40_Q_m70} shows the RMSE versus the amplitude of the noise.
Towards the left, the curves converge to the same value as the noise approaches zero.
At the higher end of the range, the curves of \MULT{} and \SRN{}
are approximately twice as steep as that of \SRD{}.
Again, \SRD{} performs uniformly superior to the rest, with \SRA{} performing second best.
In contrasts, \ADD{} performs worse than \MULT{} for a noise strength multiplier smaller than 0.2,
but better as the noise gets stronger.

\section{Summary and discussion}
\label{sec:Discussion}

The main effort of this study has been to extend the square root approach of the EnKF analysis step
to the forecast step in order to account for model noise.
Although the primary motivation is to eliminate the need for simulated, stochastic perturbations,
the core method, \SRN{}, was also found to possess several other desirable properties,
which it shares with the analysis square root update.
In particular, a formal survey on these features revealed that the symmetric square root choice for the
transform matrix can be beneficial in regards to dynamical consistency.

Yet, since it does not account for the residual noise,
\SRN{} was found to be deficient in case the noise is strong and the dynamics relatively linear.
In dealing with the residual noise, cursory experiments (not shown) suggested that an additive approach works better than
a multiplicative approach, similar to the forgetting factor of the SEIK.
This is likely a reflection of the relative performances of \ADD{} and \MULT{},
as well as the findings of \citet{whitaker2012evaluating},
which indicate that the additive approach is better suited to account for model error.
Therefore, two additive techniques were proposed to complement \SRN{}, namely \SRA{} and \SRD{}.
Adding simulated noise with no components in the ensemble subspace, \SRA{} is computationally relatively cheap
as well as intuitive.
However, it was shown to yield biased covariance updates due to the presence of cross terms.
By reintroducing dependence between the \SRN{} update and the sampled, residual noise,
\SRD{} remedies this deficiency at the cost of an additional SVD.

The utility of the noise integration methods proposed will depend on the properties
of the system under consideration.
However, \SRD{} was found to perform robustly (nearly uniformly) better than all of the other methods.
Moreover, the computationally less expensive method \SRA{} was also found to have robust performance.
These findings are further supported by omitted experiments using
fewer observations, larger observation error, and different models.

\subsection*{Future directions}
The model noise square root approach has shown significant promise on low-order models,
but has not yet been tested on realistic systems.
It is also not clear how this approach performs with more realistic forms of model error.

As discussed in Appendix B, a more shrewd choice of $\Q^{1/2}$ might improve \SRD{}.
This choice impacts $\hXi$, but not the core method, 
as shown in Appendix A section \ref{sec:Eliminating_C},
and should not be confused with the choice of $\T^f$.
While the Cholesky factor yielded worse performance than the symmetric choice,
other options should be contemplated.

\citet{nakano2013prediction} proposed a method that is distinct, yet quite similar to \SRN{},
this should be explored further, in particular with regards to the residual noise.

\acknowledgments
The authors thank Marc Bocquet for many perspectives and ideas,
some of which are considered for ongoing and future research.
Additionally, Chris Farmer and Irene Moroz have been very helpful in improving the numerics.
The work has been funded by Statoil Petroleum AS, with co-funding by the European FP7 project SANGOMA (grant no. 283580).

\appendix
\appendixtitle{The residual noise}
\subsection{The cross terms}
\label{sec:cross}
Let $\C$ be the sum of the two cross terms:
\begin{align}
	\label{eqn:Cross_1}
	\C &= \Qhh \Z\tr + \Z \Qhtr \\
	&= \PiA \Q (\I_m - \PiA) + (\I_m - \PiA) \Q \PiA \, .
\end{align}
Note that $\vspan(\Qhh \Z\tr) \subseteq \vspan(\A) \subseteq \ker(\Qhh \Z\tr)$, and therefore
$\Qhh \Z\tr$ (and its transpose) only has the eigenvalue 0.
Alternatively one can show that it is nilpotent of degree 2.
By contrast, the nature of the eigenvalues of $\C$ is quite different.
\begin{theo}[Properties of $\C$]
	\label{theo:cross_ews}
The symmetry of $\C \in \Reals^{m^2}$ implies, by the spectral theorem,
that its spectrum is real.
Suppose that $\lambda$ is a non-zero eigenvalue of $\C$,
with eigenvector $\bv = \vA + \vB$,
where	$\vA = \PiA \bv$ and $\vB = (\I_m - \PiA) \bv$.
Then (a) $\bu = \vA - \vB$ is also an eigenvector, (b) its eigenvalue is $-\lambda$,
and (c) neither $\vA$ nor $\vB$ are zero.
\end{theo}
\begin{proof}
	Note that
	\begin{alignat}{2}
		\label{eqn:CvA}
		\C \vA &= (\I_m-\PiA) \Q \vA & &\in \vspan(\A)^\perp \, , \\
		\label{eqn:CvB}
		\C \vB &=  \PiA \Q \vB    & &\in \vspan(\A) \, .
	\end{alignat}
	As
	$\C \bv = \lambda [\vA + \vB]$, \cref{eqn:CvA,eqn:CvB} imply that
	\begin{align}
		\label{eqn:CVAB2}
		\C \vA &= \lambda \vB \, , \\
		\C \vB &= \lambda \vA \, .
		\label{eqn:CVAB4}
	\end{align}
	Therefore,
	\begin{align}
		\C \bu &= \C [\vA - \vB] = \lambda \vB - \lambda \vA = -\lambda [\vA - \vB] \qedhere
		\label{eqn:Cross_F}
	\end{align}
	\Cref{eqn:CVAB2,eqn:CVAB4} can also be seen to imply (c).
\end{proof}
\begin{coro}{$\trace(\C) = 0$.}
	\label{coro:trace_zero}
	This follows from the fact that the trace of a matrix equals the sum of its eigenvalues.
\end{coro}
\begin{coro}{$\norm{\vA}_2^2 = \norm{\vB}_2^2$.}
	\label{coro:equally_sized_components}
	This follows from the fact that
	$\bv\tr \bu = (\vA + \vB)\tr (\vA - \vB) = \vA\tr\vA - \vB\tr\vB$
	should be zero by the spectral theorem.
\end{coro}
Interestingly, imaginary, skew-symmetric matrices also have the property that their eigenvalues,
all of which are real, come in positive/negative pairs.
These matrices can all be written $\M-\M\tr$ for some $\M \in \Imags^{m^2}$,
which is very reminiscent of $\C$.
However, it is not clear if these parallels can be used to prove
\cref{theo:cross_ews} because $\M-\M\tr$ only has zeros on the diagonal,
while $\C$ generally does not (by symmetry, it can be seen that this would imply $\C=0$).
Also, \cref{theo:cross_ews} depends on the fact that the cross terms are ``flanked'' by orthogonal
projection matrices, whereas there are no requirements on $\M$.

\subsection{The residual covariance matrix}
\label{sec:the_res}
The residual, $\QQh$, differs from the symmetric, positive matrix $\Z \Z\tr$
by the cross terms, $\C$. The following theorem establishes a problematic consequence.
\begin{theo}[$\QQh$ is not a covariance matrix]
	\label{theo:neg_ews}
	Provided $\C \neq 0$, the residual ``covariance'' matrix, $\QQh$, has negative eigenvalues.
\end{theo}
\begin{proof}
	Since $\C$ is symmetric, and thus orthogonally diagonalisable,
	the assumption that $\C \neq 0$ implies that $\C$ has non-zero eigenvalues.
	Let $\bv$ be the eigenvector of a non-zero eigenvalue,
	and write $\bv = \vA + \vB$, with $\vA \in \vspan(\A)$ and $\vB \in \vspan(\A)^\perp$.
	Then $\bv\tr \C \bv = \vA \tr \Q \vB \neq 0$.
	Define $\bv_{\alpha} = \vB + \alpha \vA$.
	Then:
	\begin{align}
		\bv_{\alpha}\tr \QQh \bv_{\alpha}
		&=
		\bv_{\alpha}\tr [\Z \Z\tr + \C] \bv_{\alpha}
		\\
		&=
		\vB\tr \Q \vB + 2 \alpha \vA\tr \Q \vB \, .
	\end{align}
	The second term can always be made negative, but larger in magnitude than the first,
	simply by choosing the sign of $\alpha$ and making it sufficiently large.
\end{proof}

\subsection{Eliminating the cross terms}
\label{sec:Eliminating_C}
Can the cross terms be entirely eliminated in some way?
\cref{sec:repl_one_with_two} already answered this question in the negative:
there is no particular choice of the square root of $\Q$,
inducing a choice of $\Qhh$ and $\Z$ through \cref{eqn:Qhh_def,eqn:Z_def},
that eliminates the cross terms, $\C$.

But suppose we allow changing the ensemble subspace.
For example, suppose the partition $\Q^{1/2} = \Qhh + \Z$ uses the projector onto the $N$
largest-eigenvalue eigenvectors of $\Q$ instead of $\PiA$.
It can then be shown that the cross terms are eliminated:
$\Qhh \Z\tr = 0$, and hence $\C = 0$ and $\Var(\qperp) = \Q$.
A similar situation arises in the case of the COFFEE algorithm (\cref{sec:Alternative_approaches}),
explaining why it does not have the cross term problem.
Another particular rank-$N$ square root that yields $\C = 0$
is the lower-triangular Cholesky factor of $\Q$ with the last $m-N$ columns set to zero.

Unfortunately, for general $\Q$ and $\A$, the ensemble subspace will not be that
of the rank-$N$ truncated Cholesky or eigenvalue subspace.
Therefore neither of these options can be carried out using a right-multiplying square root.

\appendixtitle{Consistency of \SRD{}}
\SRN{} ensures that \cref{eqn:AQhat_cond} is satisfied, i.e. that
\begin{align}
	\frac{1}{N-1}[ \A + \hD  ] [ \A + \hD ]\tr = \barP + \hQ \, ,
	\label{eqn:PQ_of_sqrt_upd}
\end{align}
where $(N-1) \barP = \A \A\tr$.
However, this does not imply that $\hD \hD\tr = (N-1)\hQ$.
Therefore, in reference to \SRD{}, $\hXi \hXi\tr \neq \I_m$.
Instead, the magnitudes of $\hD$ and $\hXi$ are minimised as much as possible,
as per \cref{theo:T_sym_opt}.

However, \SRD{} is designed assuming that $\hXi$ is stochastic,
with its columns drawn independently from $\NormDist(0,\I_m)$.
If this were the case, then
\SRD{} would be consistent in the sense of
\begin{align}
	\frac{1}{N-1} \Expect \left([\A + \hD + \tD] [\A + \hD + \tD]\tr \right) = \barP + \Q \, ,
	\label{eqn:SRD_consistent}
\end{align}
where the expectation is with respect to $\tXi$ and $\hXi$.
This follows from the consistency of $\q$ as defined in \cref{eqn:q_fin_def},
which has $\Var(\q) = \Q$,
because each column of $\D = \hD + \tD$ is sampled in the same manner as $\q$.

The fact that $\hD$ is in fact not stochastic, as \SRD{} assumes,
but typically of a much smaller magnitude, suggests a few possible venues for future improvement.
For example we speculate that inflating $\hXi$ by a factor larger than one,
possibly estimated in a similar fashion to \citet{dee1995line}.
The value of $\hXi$ also depends on the choice of square root for $\Qhh$.
It may therefore be a good idea to choose $\Qhh$ somewhat randomly,
so as to induce more randomness in the square root ``noise'', $\hXi$.
One way of doing so is to apply a right-multiplying rotation matrix to $\Qhh$.
Cursory experiments indicate that there may be improvements using either of the above two suggestions.

\appendixtitle{Left-multiplying formulation of \SRNb{}}
\begin{lemm}
	\label{lemm:rowspace_of_T}
	The row (and column) space of $\T^f_s = (\G^f)^{1/2}_s$ is the row space of $\A$.
\end{lemm}
\begin{proof}
	Let $\A = \U \Sig \V\tr$ be the SVD of $\A$. Then:
	\begin{align}
		\G^f &= \I_N + (N-1) \A\pinv \Q {(\A\pinv)\tr}
		\\
		&= \V \left( \I_N + (N-1) \Sig\pinv \U\tr \Q \U {(\Sig\pinv)\tr} \right) \V\tr \qedhere
	\end{align}
\end{proof}
In view of \cref{lemm:rowspace_of_T} it seems reasonable that there should be a left-multiplying update,
$\A^f = \bL \A$, such that it equals the right-multiplying update, $\A^f = \A \T^f_s$.
Although $N \ll m$ in most applications of the EnKF,
the left-multiplying update would be a lot less costly to
compute than the right-multiplying one in such cases if $N \gg m$.
The following derivation of an explicit formula for $\bL$ is very close to that of \citet{sakov2008implications},
except for the addition of \cref{eqn:Apinv_12}. \Cref{lemm:AMAk} will also be of use.
\begin{lemm}
	\label{lemm:AMAk}
	For any matrices, $\A \in \Reals^{m \times N}$, $\M \in \Reals^{m^2}$,
	and any positive integer, $k$,
	\begin{align}
		\A (\A\tr \M \A)^k = (\A \A\tr \M)^k \A \, .
		\label{eqn:AMAk}
	\end{align}
\end{lemm}

\begin{theo}[Left-multiplying transformation]
	\label{theo:left-mult}
	For any ensemble anomaly matrix, $\A \in \Reals^{m \times N}$,
	and any SPD matrix $\Q \in \Reals^{m^2}$,
	\begin{align}
		\A \T^f_s = \bL \A \, ,
		\label{eqn:right_T}
	\end{align}
	where
	\begin{align}
		\label{eqn:LT_T}
		\T^f_s &= \left(\I_N + (N-1) \A\pinv \Q (\A\pinv)\tr \right)^{1/2}_s \, ,
		\\
		\label{eqn:LT_L}
		\bL &= \left(\I_m + (N-1) \A \A\pinv \Q (\A \A\tr)\pinv\right)^{1/2} \, .
	\end{align}
	In case $N>m$, \cref{eqn:LT_L} reduces to
	\begin{align}
		\label{eqn:LT_L2}
		\bL &= \left(\I_m + (N-1) \Q (\A \A\tr)^{-1}\right)^{1/2} \, .
	\end{align}
\end{theo}
Note that $(\I_m + \A \A\pinv \Q (\A \A\tr)\pinv)$ is not a symmetric matrix.
We can nevertheless define its square root as the square root obtained from its eigendecomposition,
as was done for the symmetric square root in \cref{sec:Choice_of_square_root}.
\begin{proof}
	Assuming $\A\pinv \Q (\A\pinv)\tr$ has eigenvalues less than 1, we can express the
	square root, $(\A\pinv \Q (\A\pinv)\tr)^{1/2}$, through its Taylor expansion \citep[Th. 9.1.2]{golub1996matrix}.
	Applying \cref{eqn:Apinv_12}, followed by \cref{lemm:AMAk} with $\M = (\A \A\tr)\pinv (N-1) \Q (\A \A\tr)\pinv$,
	and \cref{eqn:Apinv_12} the other way again, one obtains \cref{eqn:LT_L}.

	If $N>m$, then $\rank(\A) = m$, unless the dynamics
	have made some of the anomalies collinear.
	Hence $\rank(\A \A\tr) = m$ and	so $\A \A\tr$ is invertible, and $\A \A\pinv = I_m$.
	Thus, \cref{eqn:LT_L} reduces to  \cref{eqn:LT_L2}.
\end{proof}

Note that the existence of a left-multiplying formulation of
the right multiplying operation $\A \mapsto \A \T^f_s$
could be used as a proof for \cref{theo:pres_mean},
because $\bL \A \ones = 0$ by the definition (\ref{eqn:A_pure}) of $\A$.
Finally, \cref{theo:indicerct_left} provides an indirect formula for $\bL$.

\begin{theo}[Indirect left-multiplying formula]
	\label{theo:indicerct_left}
	If we have already calculated the right-multiplying transform matrix $\T^f_s$, then
	the we can obtain a corresponding left-multiplying matrix, $\bL$, from:
	\begin{align}
		\bL &= \A \T^f_s \A\pinv \, .
		\label{eqn:left_mult_indir}
	\end{align}
\end{theo}
\begin{proof}
	We need to show that $\bL \A = \A \T^f_s$.
	Note that $\A\pinv \A$ is the orthogonal (and hence symmetric) projector onto the row space of $\A$,
	which \cref{lemm:rowspace_of_T} showed is also the row and column space of $\T^f_s$.
	Therefore $\T^f_s (\A\pinv \A) = \T^f_s$, and
	$\bL \A = \A \T^f_s (\A\pinv \A) = \A \T^f_s$.
\end{proof}

\bibliographystyle{ametsoc2014}
\bibliography{localrefs}

\begin{thebibliography}{60}
\providecommand{\natexlab}[1]{#1}
\providecommand{\url}[1]{\texttt{#1}}
\renewcommand{\UrlFont}{\rmfamily}
\providecommand{\urlprefix}{URL }
\expandafter\ifx\csname urlstyle\endcsname\relax
  \providecommand{\doi}[1]{doi:\discretionary{}{}{}#1}\else
  \providecommand{\doi}{doi:\discretionary{}{}{}\begingroup
  \urlstyle{rm}\Url}\fi
\providecommand{\eprint}[2][]{\url{#2}}

\bibitem[{Anderson and Moore(1979)Anderson, and Moore}]{anderson1979}
Anderson, B. D.~O., and J.~B. Moore, 1979: \textit{Optimal Filtering}.
  Prentice-Hall, Englewood Cliffs, NJ.

\bibitem[{Anderson(2001)}]{anderson2001ensemble}
Anderson, J.~L., 2001: An ensemble adjustment {K}alman filter for data
  assimilation. \textit{Monthly Weather Review}, \textbf{129~(12)}, 2884--2903.

\bibitem[{Anderson(2009)}]{anderson2009spatially}
Anderson, J.~L., 2009: Spatially and temporally varying adaptive covariance
  inflation for ensemble filters. \textit{Tellus A}, \textbf{61~(1)}, 72--83.

\bibitem[{Anderson and Anderson(1999)Anderson, and
  Anderson}]{anderson1999monte}
Anderson, J.~L., and S.~L. Anderson, 1999: A {Monte Carlo} implementation of
  the nonlinear filtering problem to produce ensemble assimilations and
  forecasts. \textit{Monthly Weather Review}, \textbf{127~(12)}, 2741--2758.

\bibitem[{Arnold et~al.(2013)Arnold, Moroz,, and Palmer}]{arnold2013stochastic}
Arnold, H.~M., I.~M. Moroz, and T.~N. Palmer, 2013: Stochastic parametrizations
  and model uncertainty in the {Lorenz'96} system. \textit{Philosophical
  Transactions of the Royal Society A: Mathematical, Physical and Engineering
  Sciences}, \textbf{371~(1991)}, 20110\,479.

\bibitem[{Ben-Israel and Greville(2003)Ben-Israel, and
  Greville}]{ben2003generalized}
Ben-Israel, A., and T.~N.~E. Greville, 2003: \textit{Generalized Inverses.
  Theory and Applications}. 2nd ed., CMS Books in Mathematics/Ouvrages de
  Math{\'e}matiques de la SMC, 15, Springer-Verlag, New York, xvi+420 pp.

\bibitem[{Berry and Harlim(2014)Berry, and Harlim}]{berry2014linear}
Berry, T., and J.~Harlim, 2014: Linear theory for filtering nonlinear
  multiscale systems with model error. \textit{Proceedings of the Royal Society
  A: Mathematical, Physical and Engineering Science}, \textbf{470~(2167)},
  20140\,168.

\bibitem[{Bishop et~al.(2001)Bishop, Etherton,, and
  Majumdar}]{bishop2001adaptive}
Bishop, C.~H., B.~J. Etherton, and S.~J. Majumdar, 2001: Adaptive sampling with
  the ensemble transform {K}alman filter. {Part I: Theoretical} aspects.
  \textit{Monthly Weather Review}, \textbf{129~(3)}, 420--436.

\bibitem[{Bocquet(2011)}]{bocquet2011ensemble}
Bocquet, M., 2011: Ensemble {K}alman filtering without the intrinsic need for
  inflation. \textit{Nonlinear Processes in Geophysics}, \textbf{18~(5)},
  735--750.

\bibitem[{Burgers et~al.(1998)Burgers, Jan~{van Leeuwen},, and
  Evensen}]{burgers1998analysis}
Burgers, G., P.~Jan~{van Leeuwen}, and G.~Evensen, 1998: Analysis scheme in the
  ensemble {K}alman filter. \textit{Monthly Weather Review}, \textbf{126~(6)},
  1719--1724.

\bibitem[{Chandrasekar et~al.(2008)Chandrasekar, Kim, Bernstein,, and
  Ridley}]{chandrasekar2008reduced}
Chandrasekar, J., I.~S. Kim, D.~S. Bernstein, and A.~J. Ridley, 2008:
  Reduced-rank unscented {K}alman filtering using {C}holesky-based
  decomposition. \textit{International Journal of Control}, \textbf{81~(11)},
  1779--1792.

\bibitem[{Chen and Oliver(2013)Chen, and Oliver}]{chen2013levenberg}
Chen, Y., and D.~S. Oliver, 2013: Levenberg--{M}arquardt forms of the iterative
  ensemble smoother for efficient history matching and uncertainty
  quantification. \textit{Computational Geosciences}, \textbf{17~(4)},
  689--703.

\bibitem[{Cotter and Reich(2012)Cotter, and Reich}]{cotter2012ensemble}
Cotter, C.~J., and S.~Reich, 2012: Ensemble filter techniques for intermittent
  data assimilation -- a survey. \textit{arXiv preprint arXiv:1208.6572}.

\bibitem[{Daley(1992)}]{daley1992estimating}
Daley, R., 1992: Estimating model-error covariances for application to
  atmospheric data assimilation. \textit{Monthly Weather Review},
  \textbf{120~(8)}, 1735--1746.

\bibitem[{Dee(1995)}]{dee1995line}
Dee, D.~P., 1995: On-line estimation of error covariance parameters for
  atmospheric data assimilation. \textit{Monthly Weather Review},
  \textbf{123~(4)}, 1128--1145.

\bibitem[{Deng et~al.(2011)Deng, Tang,, and Freeland}]{deng2011evaluation}
Deng, Z., Y.~Tang, and H.~J. Freeland, 2011: Evaluation of several model error
  schemes in the {EnKF} assimilation: Applied to {A}rgo profiles in the
  {P}acific {O}cean. \textit{Journal of Geophysical Research: Oceans
  (1978--2012)}, \textbf{116~(C9)}.

\bibitem[{Evensen(1994)}]{evensen1994sequential}
Evensen, G., 1994: Sequential data assimilation with a nonlinear
  quasi-geostrophic model using {Monte Carlo} methods to forecast error
  statistics. \textit{Journal of Geophysical Research}, \textbf{99}, 10--10.

\bibitem[{Evensen(2003)}]{evensen2003ensemble}
Evensen, G., 2003: The ensemble {K}alman filter: Theoretical formulation and
  practical implementation. \textit{Ocean Dynamics}, \textbf{53~(4)}, 343--367.

\bibitem[{Evensen(2009)}]{evensen2009ensemble}
Evensen, G., 2009: The ensemble {K}alman filter for combined state and
  parameter estimation. \textit{Control Systems, IEEE}, \textbf{29~(3)},
  83--104.

\bibitem[{Evensen and {van Leeuwen}(1996)Evensen, and {van
  Leeuwen}}]{evensen1996assimilation}
Evensen, G., and P.~J. {van Leeuwen}, 1996: Assimilation of {G}eosat altimeter
  data for the {A}gulhas current using the ensemble {K}alman filter with a
  quasigeostrophic model. \textit{Monthly Weather Review}, \textbf{124~(1)},
  85--96.

\bibitem[{Golub and {Van Loan}(1996)Golub, and {Van Loan}}]{golub1996matrix}
Golub, G.~H., and C.~F. {Van Loan}, 1996: \textit{Matrix Computations. 1996}.
  3rd ed., Johns Hopkins University, Press, Baltimore, MD, USA.

\bibitem[{Hamill and Whitaker(2005)Hamill, and Whitaker}]{hamill2005accounting}
Hamill, T.~M., and J.~S. Whitaker, 2005: Accounting for the error due to
  unresolved scales in ensemble data assimilation: A comparison of different
  approaches. \textit{Monthly Weather Review}, \textbf{133~(11)}, 3132--3147.

\bibitem[{Hanea et~al.(2007)Hanea, Velders, Segers, Verlaan,, and
  Heemink}]{hanea2007hybrid}
Hanea, R.~G., G.~J.~M. Velders, A.~J. Segers, M.~Verlaan, and A.~W. Heemink,
  2007: A hybrid {K}alman filter algorithm for large-scale atmospheric
  chemistry data assimilation. \textit{Monthly Weather Review},
  \textbf{135~(1)}, 140--151.

\bibitem[{Heemink et~al.(2001)Heemink, Verlaan,, and
  Segers}]{heemink2001variance}
Heemink, A.~W., M.~Verlaan, and A.~J. Segers, 2001: Variance reduced ensemble
  {K}alman filtering. \textit{Monthly Weather Review}, \textbf{129~(7)},
  1718--1728.

\bibitem[{Hoang et~al.(2005)Hoang, Baraille,, and
  Talagrand}]{hoang2005adaptive}
Hoang, H.~S., R.~Baraille, and O.~Talagrand, 2005: On an adaptive filter for
  altimetric data assimilation and its application to a primitive equation
  model, {MICOM}. \textit{Tellus A}, \textbf{57~(2)}, 153--170.

\bibitem[{Horn and Johnson(2013)Horn, and Johnson}]{horn2012matrix}
Horn, R.~A., and C.~R. Johnson, 2013: \textit{Matrix analysis}. 2nd ed.,
  Cambridge University Press, Cambridge, xviii+643 pp.

\bibitem[{Hoteit et~al.(2002)Hoteit, Pham,, and Blum}]{hoteit2002simplified}
Hoteit, I., D.-T. Pham, and J.~Blum, 2002: A simplified reduced order kalman
  filtering and application to altimetric data assimilation in tropical
  pacific. \textit{Journal of Marine systems}, \textbf{36~(1)}, 101--127.

\bibitem[{Hoteit et~al.(2015)Hoteit, Pham, Gharamti,, and
  Luo}]{hoteit2015mitigating}
Hoteit, I., D.-T. Pham, M.~Gharamti, and X.~Luo, 2015: Mitigating observation
  perturbation sampling errors in the stochastic {EnKF}. \textit{Monthly
  Weather Review}, in review.

\bibitem[{Houtekamer et~al.(2005)Houtekamer, Mitchell, Pellerin, Buehner,
  Charron, Spacek,, and Hansen}]{houtekamer2005atmospheric}
Houtekamer, P.~L., H.~L. Mitchell, G.~Pellerin, M.~Buehner, M.~Charron,
  L.~Spacek, and B.~Hansen, 2005: Atmospheric data assimilation with an
  ensemble {K}alman filter: Results with real observations. \textit{Monthly
  Weather Review}, \textbf{133~(3)}, 604--620.

\bibitem[{Hunt et~al.(2007)Hunt, Kostelich,, and Szunyogh}]{hunt2007efficient}
Hunt, B.~R., E.~J. Kostelich, and I.~Szunyogh, 2007: Efficient data
  assimilation for spatiotemporal chaos: A local ensemble transform {K}alman
  filter. \textit{Physica D: Nonlinear Phenomena}, \textbf{230~(1)}, 112--126.

\bibitem[{Hunt et~al.(2004)}]{hunt2004four}
Hunt, B.~R., and Coauthors, 2004: Four-dimensional ensemble {K}alman filtering.
  \textit{Tellus A}, \textbf{56~(4)}, 273--277.

\bibitem[{Jazwinski(1970)}]{jazwinski1970stochastic}
Jazwinski, A.~H., 1970: \textit{Stochastic Processes and Filtering Theory},
  Vol.~63. Academic Press.

\bibitem[{Li et~al.(2009)Li, Kalnay, Miyoshi,, and Danforth}]{li2009accounting}
Li, H., E.~Kalnay, T.~Miyoshi, and C.~M. Danforth, 2009: Accounting for model
  errors in ensemble data assimilation. \textit{Monthly Weather Review},
  \textbf{137~(10)}, 3407--3419.

\bibitem[{Livings et~al.(2008)Livings, Dance,, and
  Nichols}]{livings2008unbiased}
Livings, D.~M., S.~L. Dance, and N.~K. Nichols, 2008: Unbiased ensemble square
  root filters. \textit{Physica D: Nonlinear Phenomena}, \textbf{237~(8)},
  1021--1028.

\bibitem[{Lorenz(1963)}]{lorenz1963deterministic}
Lorenz, E.~N., 1963: Deterministic nonperiodic flow. \textit{Journal of the
  Atmospheric Sciences}, \textbf{20~(2)}, 130--141.

\bibitem[{Lorenz and Emanuel(1998)Lorenz, and Emanuel}]{lorenz1998optimal}
Lorenz, E.~N., and K.~A. Emanuel, 1998: Optimal sites for supplementary weather
  observations: Simulation with a small model. \textit{Journal of the
  Atmospheric Sciences}, \textbf{55~(3)}, 399--414.

\bibitem[{McLay et~al.(2008)McLay, Bishop,, and Reynolds}]{mclay2008evaluation}
McLay, J.~G., C.~H. Bishop, and C.~A. Reynolds, 2008: Evaluation of the
  ensemble transform analysis perturbation scheme at nrl. \textit{Monthly
  Weather Review}, \textbf{136~(3)}, 1093--1108.

\bibitem[{Mitchell and Carrassi(2014)Mitchell, and
  Carrassi}]{mitchell2014carassi}
Mitchell, L., and A.~Carrassi, 2014: Accounting for model error due to
  unresolved scales within ensemble {K}alman filtering. \textit{Quarterly
  Journal of the Royal Meteorological Society}, DOI: 10.1002/qj.2451.

\bibitem[{Nakano(2013)}]{nakano2013prediction}
Nakano, S., 2013: A prediction algorithm with a limited number of particles for
  state estimation of high-dimensional systems. \textit{Information Fusion
  (FUSION), 2013 16th International Conference on}, IEEE, 1356--1363.

\bibitem[{Nicolis(2004)}]{nicolis2004dynamics}
Nicolis, C., 2004: Dynamics of model error: The role of unresolved scales
  revisited. \textit{Journal of the Atmospheric Sciences}, \textbf{61~(14)},
  1740--1753.

\bibitem[{Oliver(2014)}]{oliver2014minimization}
Oliver, D.~S., 2014: Minimization for conditional simulation: Relationship to
  optimal transport. \textit{Journal of Computational Physics}, \textbf{265},
  1--15.

\bibitem[{Ott et~al.(2004)}]{ott2004local}
Ott, E., and Coauthors, 2004: A local ensemble {K}alman filter for atmospheric
  data assimilation. \textit{Tellus A}, \textbf{56~(5)}, 415--428.

\bibitem[{Pham(2001)}]{pham2001stochastic}
Pham, D.~T., 2001: Stochastic methods for sequential data assimilation in
  strongly nonlinear systems. \textit{Monthly Weather Review},
  \textbf{129~(5)}, 1194--1207.

\bibitem[{Rodgers(2000)}]{rodgers2000inverse}
Rodgers, C.~D., 2000: \textit{Inverse Methods for Atmospheric Sounding}. World
  Scientific.

\bibitem[{Sakov and Oke(2008{\natexlab{a}})Sakov, and
  Oke}]{sakov2008deterministic}
Sakov, P., and P.~R. Oke, 2008{\natexlab{a}}: A deterministic formulation of
  the ensemble {K}alman filter: an alternative to ensemble square root filters.
  \textit{Tellus A}, \textbf{60~(2)}, 361--371.

\bibitem[{Sakov and Oke(2008{\natexlab{b}})Sakov, and
  Oke}]{sakov2008implications}
Sakov, P., and P.~R. Oke, 2008{\natexlab{b}}: Implications of the form of the
  ensemble transformation in the ensemble square root filters. \textit{Monthly
  Weather Review}, \textbf{136~(3)}, 1042--1053.

\bibitem[{Shutts(2005)}]{shutts2005kinetic}
Shutts, G., 2005: A kinetic energy backscatter algorithm for use in ensemble
  prediction systems. \textit{Quarterly Journal of the Royal Meteorological
  Society}, \textbf{131~(612)}, 3079--3102.

\bibitem[{Slingo and Palmer(2011)Slingo, and Palmer}]{slingo2011uncertainty}
Slingo, J., and T.~Palmer, 2011: Uncertainty in weather and climate prediction.
  \textit{Philosophical Transactions of the Royal Society A: Mathematical,
  Physical and Engineering Sciences}, \textbf{369~(1956)}, 4751--4767.

\bibitem[{Tippett et~al.(2003)Tippett, Anderson, Bishop, Hamill,, and
  Whitaker}]{tippett2003ensemble}
Tippett, M.~K., J.~L. Anderson, C.~H. Bishop, T.~M. Hamill, and J.~S. Whitaker,
  2003: Ensemble square root filters. \textit{Monthly Weather Review},
  \textbf{131~(7)}, 1485--1490.

\bibitem[{{v}an Leeuwen(2009)}]{van2009particle}
{v}an Leeuwen, P.~J., 2009: Particle filtering in geophysical systems.
  \textit{Monthly Weather Review}, \textbf{137~(12)}, 4089--4114.

\bibitem[{Verlaan and Heemink(1997)Verlaan, and Heemink}]{verlaan1997tidal}
Verlaan, M., and A.~W. Heemink, 1997: Tidal flow forecasting using reduced rank
  square root filters. \textit{Stochastic Hydrology and Hydraulics},
  \textbf{11~(5)}, 349--368.

\bibitem[{Wang and Bishop(2003)Wang, and Bishop}]{wang2003comparison}
Wang, X., and C.~H. Bishop, 2003: A comparison of breeding and ensemble
  transform {K}alman filter ensemble forecast schemes. \textit{Journal of the
  Atmospheric Sciences}, \textbf{60~(9)}.

\bibitem[{Wang et~al.(2004)Wang, Bishop,, and Julier}]{wang2004better}
Wang, X., C.~H. Bishop, and S.~J. Julier, 2004: Which is better, an ensemble of
  positive-negative pairs or a centered spherical simplex ensemble?
  \textit{Monthly Weather Review}, \textbf{132~(7)}, 1590--1605.

\bibitem[{Wang et~al.(2015)Wang, Counillon,, and Bertino}]{wang2015alleviating}
Wang, Y., F.~Counillon, and L.~Bertino, 2015: Alleviating the bias induced by
  the linear analysis update with an isopycnal ocean model. \textit{Quarterly
  Journal of the Royal Meteorological Society}, in review.

\bibitem[{Whitaker et~al.(2004)Whitaker, Compo, Wei,, and
  Hamill}]{whitaker2004reanalysis}
Whitaker, J.~S., G.~P. Compo, X.~Wei, and T.~M. Hamill, 2004: Reanalysis
  without radiosondes using ensemble data assimilation. \textit{Monthly Weather
  Review}, \textbf{132~(5)}, 1190--1200.

\bibitem[{Whitaker and Hamill(2012)Whitaker, and
  Hamill}]{whitaker2012evaluating}
Whitaker, J.~S., and T.~M. Hamill, 2012: Evaluating methods to account for
  system errors in ensemble data assimilation. \textit{Monthly Weather Review},
  \textbf{140~(9)}, 3078--3089.

\bibitem[{Whitaker et~al.(2008)Whitaker, Hamill, Wei, Song,, and
  Toth}]{whitaker2008ensemble}
Whitaker, J.~S., T.~M. Hamill, X.~Wei, Y.~Song, and Z.~Toth, 2008: Ensemble
  data assimilation with the {NCEP} global forecast system. \textit{Monthly
  Weather Review}, \textbf{136~(2)}, 463--482.

\bibitem[{Wunsch(2006)}]{wunsch2006discrete}
Wunsch, C., 2006: \textit{Discrete Inverse and State Estimation Problems: With
  Geophysical Fluid Applications}. Cambridge University Press.

\bibitem[{Zhang et~al.(2004)Zhang, Snyder,, and Sun}]{zhang2004impacts}
Zhang, F., C.~Snyder, and J.~Sun, 2004: Impacts of initial estimate and
  observation availability on convective-scale data assimilation with an
  ensemble {K}alman filter. \textit{Monthly Weather Review}, \textbf{132~(5)},
  1238--1253.

\bibitem[{Zupanski and Zupanski(2006)Zupanski, and
  Zupanski}]{zupanski2006model}
Zupanski, D., and M.~Zupanski, 2006: Model error estimation employing an
  ensemble data assimilation approach. \textit{Monthly Weather Review},
  \textbf{134~(5)}, 1337--1354.

\end{thebibliography}

\setlength{\tabcolsep}{3.5pt}
\renewcommand{\arraystretch}{1.0}
\begin{table}[h]
\caption{Inflation factors used in benchmark experiments.
Reads from left to right, corresponding to the abscissa of the plotted data series.}\label{t1}
\begin{center}
\begin{tabular}{lrrrrrrrrrr}
\topline
Fig. &  & &  \multicolumn{6}{c}{Post-analysis inflation} & & \\
\midline
\ref{fig:LA_m6}
& & & \multicolumn{6}{c}{None} & & \\
\ref{fig:L40_N_m70}
& 1.25 & 1.22 & 1.19 & 1.15 & 1.13 & 1.12 & 1.10 & 1.03 & 1.00 & 1.00 \\
\ref{fig:L40_dtObs_m70}
& 1.13 & 1.25 & 1.30 & 1.35 & 1.43 & 1.50 & 1.57 & 1.65 & 1.70 \\
\ref{fig:L40_Q_m70}
& 1.02 & 1.02 & 1.02 & 1.03 & 1.04 & 1.05 & 1.07 & 1.09 & 1.13 & \ldots \\
& 1.17 & 1.21 & 1.31 \\
\botline
\end{tabular}
\end{center}
\end{table}

\end{document}